%% file: 0-paper.tex
\documentclass[conference]{IEEEtran}
\IEEEoverridecommandlockouts
\usepackage{cite}
\usepackage{amsmath}
\usepackage{amssymb}
\usepackage{amsfonts}

\usepackage[ruled, linesnumbered, vlined]{algorithm2e}

\usepackage{array}
\usepackage{url}
\usepackage{float}
\usepackage{graphicx}
\graphicspath{{../}{./}}
\usepackage{multirow}
\usepackage{subfigure}
\usepackage{multicol}
\usepackage{tabu}
\usepackage{listings}
\usepackage{makecell}
\usepackage{enumerate}
\usepackage{color}
\usepackage[english]{babel}
\usepackage{threeparttable}
\usepackage{upgreek}
\usepackage{booktabs}
\usepackage{xcolor}
\usepackage{amsthm}
\usepackage{enumitem}
\usepackage{blindtext}
\usepackage{textcomp}
\usepackage{bm}
\usepackage{xspace}
\usepackage{colortbl} % For coloring table cells
\usepackage{comment}
\usepackage[normalem]{ulem}
\usepackage[textsize=tiny,textwidth=0.6in]{todonotes}

\usepackage{tikz}
\usetikzlibrary{matrix, arrows.meta}
\usetikzlibrary{calc}

\newcommand{\ie}{\emph{i.e.,}\xspace}
\newcommand{\eg}{\emph{e.g.,}\xspace}
\newcommand{\etc}{\emph{etc.}\xspace}

\newcommand{\MyPara}[1]{\vspace{0.1em}\noindent\textit{\textbf{#1}}}
\definecolor{dkgreen}{rgb}{0,0.6,0}
\definecolor{gray}{rgb}{0.5,0.5,0.5}
\definecolor{mauve}{rgb}{0.58,0,0.82}

\newif\ifappendix
\appendixtrue
\newcommand{\apx}[2]{\ifappendix #1\else #2\fi}

\newif\ifshowchanges
\showchangesfalse
\ifshowchanges
\newcommand{\rev}[1] {{\textcolor{cyan}{{#1}}}}
\newcommand{\del}[1]{\textcolor{red}{\sout{#1}}}
\else
\newcommand{\rev}[1]{#1}
\newcommand{\del}[1]{}
\fi

\newcommand{\mysection}[1]{\section{#1}}
\newcommand{\mysubsection}[1]{\subsection{#1}}
\newcommand{\mysubsubsection}[1]{\subsubsection{\normalsize #1}}

\newcommand{\SYS}{\texttt{Lizard}\xspace}
\newcommand{\tool}{\texttt{Lizard}\xspace}

\theoremstyle{plain}

\newtheorem{proposition}{Proposition}
\newcommand{\squishlist}{
   \begin{list}{$\bullet$}
    { \setlength{\itemsep}{0pt}      \setlength{\parsep}{3pt}
      \setlength{\topsep}{3pt}       \setlength{\partopsep}{0pt}
      \setlength{\leftmargin}{1.0em} \setlength{\labelwidth}{1em}
      \setlength{\labelsep}{0.5em} } }
\newcommand{\squishend}{
    \end{list}  }

\begin{document}

\title{\SYS: Bandwidth-Adaptive Real-Time Video Analytics through Content-Aware Packet Discarding at Last-Mile Edge Routers
}
\author{
\IEEEauthorblockN{
Shan Yu\IEEEauthorrefmark{1},
Yu Chen\IEEEauthorrefmark{2},
Yifan Qiao\IEEEauthorrefmark{3},
Sheng Zhang\IEEEauthorrefmark{4},
Ravi Netravali\IEEEauthorrefmark{5},
Harry Xu\IEEEauthorrefmark{1}
}

\IEEEauthorblockA{\IEEEauthorrefmark{1}
University of California, Los Angeles\\
Email: \{shanyu1, harryxu\}@cs.ucla.edu
}

\IEEEauthorblockA{\IEEEauthorrefmark{2}
Renmin University of China\\
Email: chenyu97@ruc.edu.cn
}

\IEEEauthorblockA{\IEEEauthorrefmark{3}
University of California, Berkeley\\
Email: yifanqiao@berkeley.edu
}

\IEEEauthorblockA{\IEEEauthorrefmark{4}
Nanjing University\\
Email: sheng@nju.edu.cn
}

\IEEEauthorblockA{\IEEEauthorrefmark{5}
Princeton University\\
Email: rnetravali@cs.princeton.edu
}

% \vspace{-2em}
}

\maketitle

\begin{abstract}
The timeliness and accuracy of edge-based video analytics can be hindered by drastic reductions in available bandwidth (ABW) at \emph{last-mile edge routers}, causing prolonged queuing delays.
This work proposes \SYS, a system that leverages \emph{video-content-aware packet discarding} to mitigate the negative effects of drastic ABW degradation that may frequently occur at a last-mile edge router by judiciously discarding packets that contain frame blocks less important to the analytics at the destination.
To achieve this, we first devise a frame-block-aware RTP header extension to effectively decouple packet dependencies to encode frame blocks.
Second, \SYS uses a priority-based feedback mechanism that dynamically evaluates packet priorities based on relative accuracy impacts.
Third, we develop an adaptive phase-transition-based packet discarding strategy at the router to discard packets that represent unimportant blocks.
Our evaluation of \SYS shows improvements over existing methods are substantial: 53.2\% reduction in latency and 27.1\% increase in analysis accuracy.
\end{abstract}


\input{1-introduction}
\input{3-motivation}

\input{4-design}

\input{6-evaluation}

\input{2-relatedwork}

\input{7-conclusion}

\bibliographystyle{IEEEtran}
\bibliography{reference}

\ifappendix
\appendices
\input{Appendix_1}
\fi

\end{document}

%% file: 1-introduction.tex
\mysection{Introduction
\label{sec:introduction}}
%\textbf{[Background]}
Recent advances in computer vision \cite{voulodimos2018deep, szeliski2022computer, granlund2013signal, weinstein2018computer} have transformed video analytics across various applications. 
Driven by the rapid proliferation of IoT devices and smart city infrastructures, edge-based video analytics has become increasingly important for supporting latency-sensitive and scalable vision services\cite{badidi2023opportunities, hu2023edge, gong2025survey}. 
Edge computing \cite{cao2020overview, mao2017survey, chen2019deep, shi2016edge, khan2019edge} processes video data near the source, enhancing privacy \cite{wang2017scalable}, network efficiency \cite{canel2019scaling}, and reducing latency \cite{yi2017lavea, yang2019edge, chen2023octopus}.
This is crucial for edge-based, real-time analytics \cite{liu2018edgeeye, li2020reducto} in applications like autonomous driving \cite{lin2022low, ni2023cellfusion}, drone navigation \cite{qu2021dronecoconet}, robotic perception and vision-based control  \cite{tahir2025edge, grover2026embodied, williams2025lite}, and augmented reality \cite{xie2020interactive, liu2019edge}, demanding low latency and high accuracy \cite{ananthanarayanan2017real, guo2021crossroi} even in poor network conditions. 

To meet the stringent latency requirements in a dynamic network environment, two main categories of work have been conducted. 
At the network level, real-time communication (RTC) technologies \cite{loreto2014real} leverage the real-time transport protocol (RTP) \cite{frederick2003rtp} for reliable transmission and real-time control protocol (RTCP) \cite{sarker2021rtp} for quality monitoring and adaptive responses to bandwidth fluctuations.
%On one hand, at the network level, real-time communication (RTC) technologies \cite{loreto2014real, rahaman2015survey} use the real-time transport protocol (RTP) \cite{frederick2003rtp, schulzrinne2003rfc3550} for reliable data transmission and the real-time control protocol (RTCP) \cite{sarker2021rtp, novotny2008large} for quality monitoring and real-time response to network bandwidth fluctuation. 
Google congestion control (GCC) \cite{carlucci2016analysis} then optimizes bandwidth usage by dynamically adjusting transmission rates for video data. % manages bandwidth to enable efficient video streaming for low-latency edge-based real-time video analytics.
%Google congestion control (GCC) \cite{carlucci2016analysis} was proposed to manage bandwidth dynamically, ensuring efficient video streaming and supporting high-performance edge-based real-time video analytics in latency-sensitive environments.
%\shan{We may consider shortening or even omitting the discussion of RTC in the Intro, as it is not central to our main storyline, while requiring large space to explain clearly.}
At the application level, a large body of work \cite{wang2020joint, zhang2021adaptive, li2020reducto, xie2019source, du2022accmpeg, canel2019scaling, yuan2022infi, han2015deep, yuan2023packetgame, yuan2023accdecoder, mi2024accelerated} aims to optimize the network efficiency of video analytics from different vantage points in a pipeline, \eg at edge devices~\cite{li2020reducto, wang2020joint, zhang2021adaptive},  edge servers~\cite{gemel-nsdi23}, and cloud servers~\cite{canel2019scaling, yuan2022infi, yuan2023accdecoder}. 

\MyPara{Last-Mile Edge Routers.} Despite these efforts, the performance of edge-based, real-time video analytics can still be severely degraded by fluctuating network capacity and available bandwidth (ABW). As illustrated in Figure~\ref{fig:Motivation-Background},
%\yq{``user'' in the figure is a bit vague. It is the user who runs video analytics queries or the user of uplink apps? I assume the former but not sure.},
this fragility arises because the cameras used for analytics typically connect to edge servers (where analytics is performed) via \emph{last-mile routers}, such as home gateways~\cite{sundaresan2011broadband, hoiland2018cake}, wireless APs~\cite{bhartia2017measurement, meng2022achieving, xu2020understanding}, and mobile edge gateways~\cite{hu2023edge, intel_upf_lowlatency, crugnola2025latency}. These routers are often the narrowest link along the path and, as such, create a bottleneck effect that dominates the end-to-end performance~\cite{sundaresan2011broadband, undaresan2013web, xu2021edge, imc2021cloudy, shao2022task, xu2022tutti}. 

Congestion at last-mile routers is extremely common\del{, as these}\rev{. These} routers must contend with not only the inherent access-link variability~\cite{sundaresan2016home, sharma2023wifi}, but also traffic spikes from competing data-dense cross traffic, \eg cloud backups~\cite{ramesh2022secured, gupta2021review}, AI model aggregation~\cite{qi2024model, deng2021fair}, and live video broadcasting~\cite{ayman2024design, jiang2021survey}. Measurements from a residential broadband video streaming trace~\cite{broadbandvideo2010} reveal that the last-mile edge routers can suffer from severe instability, with sudden ABW drops of more than 50\% occurring at a frequency of \del{2.7\%}\rev{2.68\%} on broadband access links and \del{4.3\%}\rev{4.35\%} on home wireless networks.

%\del{Indeed, our analysis of a public video streaming dataset~\cite{Oboe} \hx{Change} reveals that last-mile edge routers can experience \textbf{sudden ABW drops} of over 50\% with a frequency of 1.04\%} (\S\ref{subsec:motivation}). 
%\hx{do we have any idea of how many of bottleneck routers are also the last mile routers?}
%each lasting an average of over 9 seconds (\S\ref{subsec:motivation}). 
%\rn{does the 9 seconds matter? wouldn't things like GCC be able to deal with that?} 
%\yu{yes, you are right. It does not matter in our considered problem. I will remove it.}
This results in substantial queuing delays and packet drops at the last mile, which cascade through the path and lead to severe reductions in analytics accuracy (by up to 30\%), while also causing violations of service-level agreements (SLAs).~\cite{yan2022resource, xu2024adaptive}

% The effect is substantial queuing delays and packet drops at the last-mile, propagating through the path and ultimately, leading to severe reductions in analytics accuracy (by up to 30\%) and violations of service-level agreements (SLAs)~\cite{yan2022resource, xu2024adaptive}. 
%\rn{I reworked the para above; please make sure I didn't break something :)}

\begin{figure}
    \centering
    \includegraphics[width=0.95\linewidth]{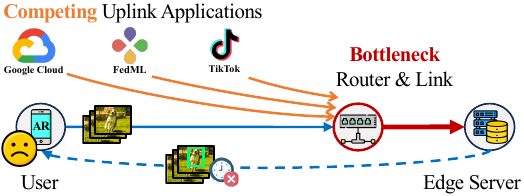}
    \vspace{0mm}
    \caption{Diverse competing uplink applications can cause a sudden drop in ABW at a last-mile edge router and link, leading to high latency and reduced accuracy in real-time video analytics on the edge.}
    \label{fig:Motivation-Background}
\end{figure}

Existing optimization techniques are incapable of preserving analytics accuracy in the presence of sudden ABW reductions at edge routers. Network-level techniques perform congestion control \cite{loreto2014real, rahaman2015survey} or proactively drop packets \cite{ryu2004advances, hollot2002analysis, lin1997dynamics, floyd1993random, patil2019towards, sharma2014controlling, liu2018adaptive, carlucci2018controlling, banchs2002distributed, hamadneh2012dynamic} at routers in a way that is agnostic to packet content (\ie frame content), and hence can drop important packets (\eg containing content that is critical to the analytics accuracy), leading to accuracy loss.
Retransmission techniques~\cite{liu2018adaptive, carlucci2018controlling} can be used to recover lost packets, but they introduce significant delays, disrupting timely processing in real-time applications. 
Moreover, application optimizations for video analytics are conducted either at the client (device) or at the (edge) server, in a way that is either agnostic to network conditions or exhibits lagged reactions to network fluctuations.
%(as analyzed in \S\ref{subsubsec:theoreticalAnalysis}).
Severe accuracy loss can still result from sudden ABW drops even if these optimizations are enabled (see \S\ref{subsec:motivation2}).

\MyPara{Insight.} The main question this paper asks is: can we design an in-network mechanism that drops packets at last-mile edge routers in a \emph{video-content-aware} manner when they encounter bandwidth reductions?  In other words, can we prioritize packets based on the impact of the content they contain on the analytics accuracy? We answer this question by discussing the design and implementation of \SYS, a system that can prioritize packets and judiciously drop packets directly at an edge router (1) based on their importance rankings and (2) in response to the severity of the ABW drop. 
While this idea appears straightforward, evaluating each packet's importance faces the following three challenges. 

The first challenge is \emph{how to make packets aware of frame content}. 
Enabling content-aware packet dropping requires a packetization technique that can (1) minimize the dependencies between packets and (2) make packets distinguishable in terms of their importance to the analytics.  Conventional techniques packetize the encoded stream of a video sequentially into packets, which may introduce a large number of dependencies\textemdash the content of a frame spans multiple consecutive packets, making it hard to recover the frame if one or multiple of these packets are dropped. The other problem is standard packetization are agnostic to frame content, making it difficult to determine their individual importance.

%For this challenge, we use a \emph{block-based} frame encoding algorithm~\cite{du2022accmpeg, shi2023adapyramid, guan2019pano, chen2024tilesr}, which divides each frame into blocks and each block undergoes independent I/P-frame encoding \cite{pourreza2021extending, turaga2001p} and network packetization. 
To address this challenge, we customize a \emph{block-based} frame encoding algorithm \cite{du2022accmpeg, shi2023adapyramid, guan2019pano, chen2024tilesr}, which divides each frame into multiple blocks. 
\rev{Each block is independently encoded with a standard lossy codec (H.264 in our implementation) as an I/P-frame  \cite{pourreza2021extending, turaga2001p} sequence; the encoded bitstream is then packetized directly into RTP packets without further transcoding.}
Our key observation is that, for a given camera feed, different blocks of a frame capture various areas of a location and often have varying levels of importance for the analytics (see Figure~\ref{fig:Motivation-2}). 
This allows us to treat these blocks as separate packets, making it easier to assess their individual impact on analytics accuracy. \tool embeds this block information into packets by using the RTP header extension mechanism~\cite{gao2019rtp}, enabling edge routers to perform content-aware dropping while remaining fully compliant with the RTP standard (ensuring valid transmission through existing networks)~\cite{frederick2003rtp, schulzrinne2003rfc3550,sarker2021rtp}. 

The second challenge is \emph{how to assess packet importance}. The importance of each packet represents its impact on the analytics accuracy and hence must be assessed in the context of the analytics. \SYS establishes a feedback loop that uses the accuracy change in the performed analytics to compute the importance of each block, which may vary over time  (\eg due to scene changes). % \rn{does the importance of a packet stay stable over time?} 
This accuracy change is evaluated with a metric referred to as \emph{relative accuracy impact} (RAI), which measures the accuracy improvement gained by retaining a block as compared to losing it. We establish our feedback loop by piggybacking on Real-Time Control Protocol (RTCP) \cite{sarker2021rtp, novotny2008large}, which complements RTP by providing critical feedback on metrics such as packet loss and jitter, facilitating real-time adjustments to improve user experience.
Among the various RTCP types, we utilize application-specific (APP) packets \cite{montagud2012enhanced}, allowing customized feedback tailored to specific application requirements. 
Specifically, we calculate the RAI for each block on the server and pass it as feedback using a carefully designed APP RTCP structure (see \S\ref{subsec:priorityFeedback}).
% \del{It is then written into the packets that encode the block for future use.} 
The feedback guides future packetization at the client.

The third challenge is \emph{how to discard packets on edge routers}. We design a phase-based approach that runs on a last-mile edge router to identify ABW drop and discard packets.
Our algorithm performs phase transitions based on network conditions, ensuring timely packet dropping.
% Our algorithm performs phase transitions based on the frame structure information to determine the optimal packet dropping timing.
It further adapts the dropping threshold (\ie an importance value such that packets whose importance is below this value will be dropped) to the ABW, effectively discarding unimportant packets by comparing each packet's importance with the threshold.

\MyPara{Compatibility and Deployability.} 
\SYS is engineered for drop-in deployability in contemporary edge routers. It confines changes to RTP/RTCP header fields, preserving wire compatibility with existing infrastructure. Per-packet analytics (\eg block structure and importance) are encoded as compact, standards-compliant RTP/RTCP header extensions that are shipped in plaintext (even with encrypted variants like SRTP), allowing transparent carriage over today’s edge networks and access at routers. Further, the importance-aware dropping algorithm is realizable at edge routers \emph{without per-flow state}: endpoints use analytics information to annotate packets with importance/feedback in the headers, and \SYS-enabled edge routers enforce discard policies by inspecting only these fields and performing simple checks. 
\del{Many modern edge routers often do maintain per-flow state for NAT, security, QoS, \etc For those routers that already support per-flow queues, \SYS 
performs fair packet dropping to ensure that the accuracy drops are similar for different video flows; otherwise, \SYS drops packets in a flow-agnostic manner to maximize the average accuracy across all flows. }
\rev{We discuss multi-flow video support in \S\ref{subsec:packetDiscarding}.}

% While generally usable on any router, \SYS targets packet filtering at modern last-mile edge routers to enable immediate deployment. The majority of last-mile edge routers (\eg home gateways and wireless APs) are either Linux-based or operate Linux-derived firmware (such as OpenWRT~\cite{Anurag,openwrt,PeterWeidenbach}), which permits execution of the eBPF programs~\cite{vieira2020fast,abc} leveraged by \SYS. Looking ahead, next-generation 5G/6G infrastructure is increasingly built on modern Linux platforms with full eBPF support~\cite{foukas2023taking}, making \SYS readily deployable on future last-mile devices.  The current implementation of \SYS is also compatible with any secure protocol such as SRTP where packet headers are not encrypted. Encrypting headers (\eg RFC 6904) would prevent \SYS from reading header information from packets; we leave the support for such protocols to future work. 

While generally applicable to any router, \SYS targets packet filtering at modern last-mile edge routers for immediate deployment. Most such routers, including home gateways and wireless APs, are Linux-based or use Linux-derived firmware such as OpenWRT~\cite{Anurag,openwrt,PeterWeidenbach}, enabling the eBPF programs~\cite{vieira2020fast,abc} used by \SYS. Emerging 5G/6G infrastructure is likewise increasingly built on Linux platforms with eBPF support~\cite{foukas2023taking}, making \SYS readily deployable on future last-mile devices. 
\SYS is also compatible with secure protocols such as SRTP~\cite{rfc3711} when packet headers remain unencrypted. For protocols such as RFC 6904~\cite{rfc6904}, where header extensions are encrypted separately, users can enable \SYS by sharing only the header-extension key with their last-mile edge router.

\MyPara{Results.}
We have evaluated \SYS on 2 real-world network bandwidth datasets and 4 real-world video datasets, including both surveillance and moving camera videos.
Our results show that \SYS outperforms other baselines (both network-level ones like GCC-only \cite{carlucci2016analysis} and CoDel~\cite{sharma2014controlling}, and application-level ones like EAAR \cite{liu2019edge}), 
%\rn{maybe list the baselines?}
achieving a reduction in the 99.5th percentile frame delay by 53.2\% and an improvement in accuracy by up to 27.1\%, with negligible overhead.

%\vspace{3pt}
%\noindent {\em This work does not raise any ethical issues.}

%% file: 3-motivation.tex
\mysection{Background and Motivation
\label{sec:motivation}}
%\yuRevise{In this section, we first introduce the drastic ABW degradation at the bottleneck router in edge networks. 
%Next, we show the limitations of strawman approaches. 
%\shanRevise{Next, we demonstrate why the strawman approaches fail to adapt to the drastic ABW degradation issue.}}
%In this section, we first introduce the background of real-time video analytics.
%Next, we analyze the issue of drastic available bandwidth degradation at bottleneck routers for a typical edge-based video analytics pipeline as illustrated in Figure \ref{fig:Motivation-0}.
%Finally, we discuss the opportunities with frame-aware packet discarding.
\vspace{.5em}
\mysubsection{Why Last-Mile Edge Routers?}

Last-mile edge routers are the access segment linking end-user devices (\eg home cameras, IoT devices, drones, smartphones, \etc) to the first Internet Service Provider (ISP)’s aggregation network, typically through home gateways~\cite{sundaresan2011broadband, hoiland2018cake}, cellular base stations~\cite{xu2022tutti}, and fixed broadband access equipment (DSLAMs, CMTSs, OLTs) \cite{sundaresan2011broadband, bajpai2017lastmile}.

Prior work identifies the last mile as the typical bottleneck and a dominant determinant of end-to-end performance. For example, Cloudy~\cite{imc2021cloudy} reports that the last-mile wireless link accounts for more than 50\% of the total end-to-end latency to the nearest cloud. Edge platform measurements~\cite{xu2021edge} further show that under 5G the first three hops contribute up to 98\% of RTT to the nearest edge and 82.2\% to the nearest cloud.
% shares 69.8\%, 89.7\%, 97.9\% of the end to end network delay to the nearest edge (5-12 hops), and 47.5\%, 74.8\%, 82.2\% to the nearest cloud (10-16 hops) under WiFi, LTE and 5G respectively. 
Web bottleneck studies show that last-mile latency strongly governs page loads: adding just 10ms in the last mile can inflate page load times by hundreds of milliseconds~\cite{undaresan2013web}. For edge video analytics, Tutti~\cite{xu2022tutti} identifies the radio access network, the first hop in 5G, as a persistent bottleneck. In edge deployments the effect is pronounced: sustained analytics streams face limited uplink capacity and cross-traffic, leading to queuing delays, bandwidth collapses, and degraded task accuracy~\cite{shao2022task,pi2024pib,xu2021edge}.

The effect is also evident in our trace analysis. Table \ref{tab:ABW_drop_frequency} shows the frequency of ABW degradation across real-world datasets \cite{Oboe, FCC18, Ghent, broadbandvideo2010} measured in last-mile networks. While the Oboe~\cite{Oboe} dataset measures the end-to-end available bandwidth for video streaming, prior work has shown that such variability is primarily driven by the last-mile routers in both residential broadband and mobile networks~\cite{sundaresan2011broadband, bajpai2017lastmile, imc2021cloudy}, making Oboe a reasonable approximation for last-mile bandwidth dynamics. 
The frequency of 50\% ABW degradation is 1.04\%, 2.16\%, 5.22\%, 2.68\%, and 4.35\% for Oboe \cite{Oboe}, FCC18 \cite{FCC18}, Ghent \cite{Ghent}, VSBroadband~\cite{broadbandvideo2010}, and VSWireless~\cite{broadbandvideo2010}, respectively.
Notably, VSWireless experiences an average of 13 severe ABW drops (over 50\%) within 5 minutes, underscoring the highly unstable nature of in-home wireless networks.
%Such ABW reductions can cause significant frame delays, harming analytics accuracy.

Beyond being the narrowest and most congestion-prone segment of the network~\cite{sundaresan2011broadband,bajpai2017lastmile,imc2021cloudy}, the last mile is becoming increasingly programmable: 
many home gateways run Linux/OpenWrt~\cite{openwrt} and expose eBPF/XDP traffic shaping~\cite{xu2021edge, xu2022tutti}. 
On the broadband side, operators are transitioning from closed DSLAMs/CMTS designs to virtualized BNGs~\cite{itu_vBNG2025,cisco_cnBNG} and open disaggregated access (\eg vBNGs~\cite{itu_vBNG2025}, DT Access 4.0~\cite{dt_access4}, ONF SEBA~\cite{onf_seba}). 
Wireless edge platforms such as OAI~\cite{oai} and srsRAN~\cite{srsran} expose similar programmability in software-defined RANs, while emerging vRAN systems like Atlas~\cite{atlas2023} extend this flexibility to 5G/6G deployments. 
SD-Fabric~\cite{onf_sdfabric,opennetworkingUsingProgrammable}, Janus~\cite{foukas2023taking}, and P4-UPF~\cite{macdavid2021p4upf,bose2021leveraging,hip4upf@atc24} demonstrate that the last-mile data plane can be made programmable by combining P4-programmable switches with Linux servers running eBPF/XDP, enabling software-defined control at the network bottleneck. As a result, last-mile routers serve as both the performance-critical link and a natural programmable control point for adaptive resource management in edge-based video analytics.

\begin{table}[t!]
\caption{Frequency of drastic ABW degradation across different real-world bandwidth datasets.}
\centering
\footnotesize
\begin{tabular}{|c|c|c|c|c|c|}
\hline
\textbf{\begin{tabular}[c]{@{}c@{}}ABW Drop Ratio\end{tabular}} & \textbf{30\%} & \textbf{40\%} & \textbf{50\%} & \textbf{60\%} & \textbf{70\%} \\ \hline
\textbf{Oboe}~\cite{Oboe}       & 4.35\%  & 2.14\%  & 1.04\%  & 0.45\%  & 0.14\%  \\ \hline
\textbf{FCC18}~\cite{FCC18}     & 4.70\%  & 3.14\%  & 2.16\%  & 1.61\%  & 1.18\%  \\ \hline
\textbf{Ghent}~\cite{Ghent}     & 14.83\% & 8.97\%  & 5.22\%  & 3.46\%  & 2.17\%  \\ \hline
\textbf{VSBroadband}~\cite{broadbandvideo2010} & 8.70\%  & 4.35\%  & 2.68\%  & 1.67\%  & 1.00\%  \\ \hline
\textbf{VSWireless}~\cite{broadbandvideo2010}  & 10.03\% & 6.69\%  & 4.35\%  & 2.34\%  & 1.00\%  \\ \hline
\end{tabular}
\vspace{2mm}
\label{tab:ABW_drop_frequency}
\vspace{1mm}
\end{table}

\mysubsection{Why In-Network Packet Dropping?\label{subsec:motivation2}}
To motivate \SYS's in-network packet dropping technique, we first ask if existing frame filtering mechanisms can already mitigate the impact of ABW degradation on real-time video analytics. To this end, two strawman approaches emerge.

\begin{figure}[!h]
    \vspace{0mm}
    \centering
    \subfigbottomskip=0mm 
    \subfigure[Decreasing Accuracy.]{
        \includegraphics[width=0.47\linewidth]{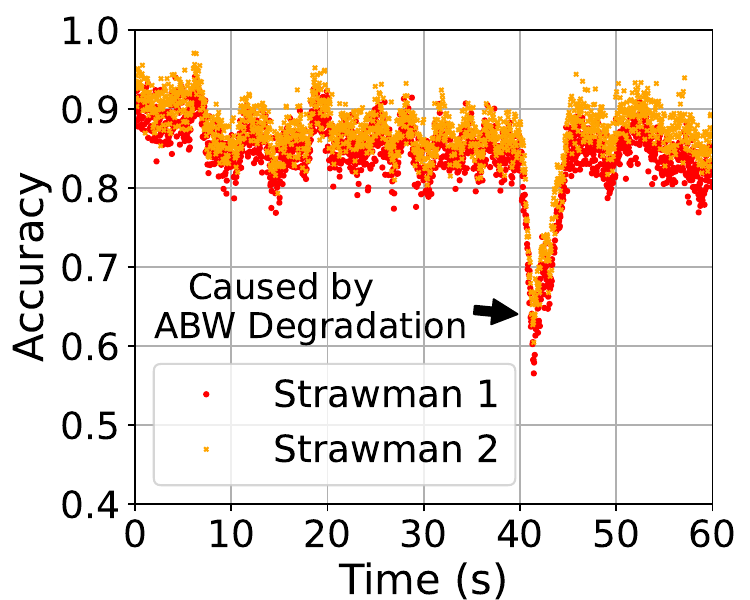}
        \label{fig:Motivation_filtering_methods_3}
    }\hfill
    \subfigure[Increasing Frame Delay.]{
        \includegraphics[width=0.47\linewidth]{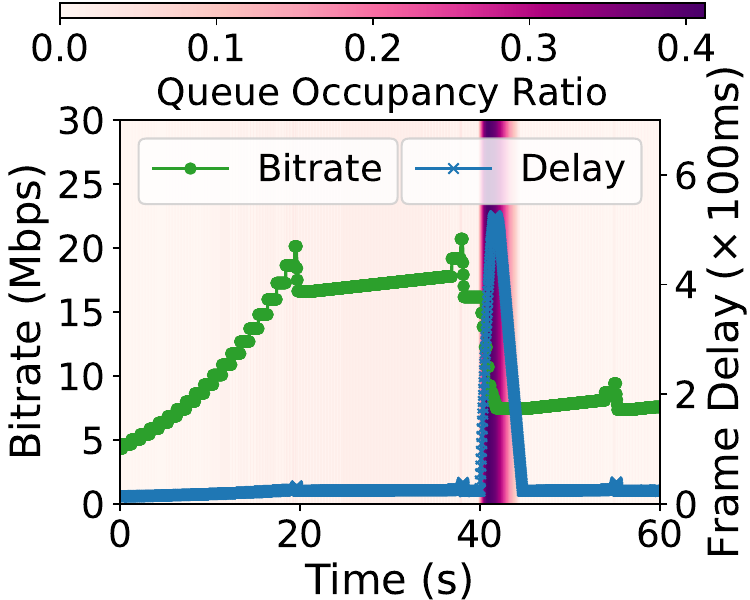}
        \label{fig:Motivation_filtering_methods_4}
    }
    \vspace{0mm}
    \caption{Strawman approaches for addressing the drastic ABW degradation in real-time video analytics on the edge.}
    \label{fig:Motivation_filtering_methods}
    \vspace{0mm}
\end{figure}

\MyPara{Strawman 1: Proactive Client-Side Filtering.}
% The first strawman approach is to proactively reduce the amount of data transmitted from the client side. Systems like Reducto~\cite{li2020reducto} and Glimpse~\cite{chen2015glimpse} use client-side computation to filter frames based on low-level pixel or edge features, \rev{} effectively reducing network consumption and improving bandwidth efficiency without severely affecting analytics accuracy. 
The first strawman approach reduces network usage by limiting data sent from the client. Reducto~\cite{li2020reducto} and Glimpse~\cite{chen2015glimpse} filter frames using low-level pixel or edge features, while \rev{offload shaping~\cite{hu2015offload, iyengar2023offload} and Abate~\cite{zhou2023bandwidth} similarly reduce the quality of unimportant regions before transmission.} These techniques effectively reduce network consumption and improve bandwidth efficiency without severely affecting analytics accuracy.

%However, these methods do not account for dynamic network conditions. When network congestion occurs, congestion control mechanisms further reduce the transmission bitrate, leading significant accuracy loss.

\MyPara{Strawman 2: Server-Guided Bitrate Reaction.}
The second strawman is to passively react to network condition changes by adjusting the video quality based on server-side feedback. Systems like DDS~\cite{du2022accmpeg, du2020server, shi2023adapyramid} leverage server-side video analytics and the network queue status to guide the client in using different encoding quality. 
%Specifically, different regions of a frame are encoded at varying quality levels based on their significance to analytics accuracy, and the client adjusts encoding quality in response to server-monitored queuing delays.

%However, this method suffers from response lag. When available bandwidth (ABW) drops suddenly, the client takes time to lower the bitrate, leading to extended queuing delays as high-bitrate packets continue to be sent. Since bitrate adjustments rely on feedback from the server—where bottleneck queuing delays are reported back to the client—the approach is inherently reactive. This delay worsens frame latency under abrupt ABW reductions, and also results in the accuracy drop of the real-time video analytics.

We implement and evaluate the two strawman approaches on the YOLOX \cite{ge2021yolox} and the Jacksonhole datasets \cite{jackson}.
Figure~\ref{fig:Motivation_filtering_methods_3} shows the accuracy changes of the two strawman approaches during the abrupt ABW reduction.
When ABW drops from 20 to 10Mbps at the 40th second, the average accuracy of the two approaches decreases from 0.87 to 0.60, and from 0.84 to 0.56, respectively. This degradation is primarily due to increased queuing delays and frame misalignment caused by network congestion.

Though Reducto-like (strawman 1) approaches can reduce the bandwidth consumption by filtering out less important data, they remain vulnerable to ABW reductions. Regardless of the bandwidth consumed by video analytics, competing applications (\eg cloud backups \cite{ramesh2022secured}, live video streaming \cite{ramesh2022secured}, \etc) continuously adjust the amounts of sent data to saturate the network~\cite{grover2022rate, balador2022survey, lorincz2021comprehensive}), causing ABW reductions to impact all applications. %As a result, video analytics face even lower bandwidth availability, leading to further accuracy degradation.
Server-side feedback approaches (strawman 2) introduce a significant lag in adapting to network changes. As shown in Figure~\ref{fig:Motivation-0}, queuing delays at the last-mile router must first propagate to the server, which then instructs the client to adjust its bitrate. This reactive process leads to delayed bitrate adaptation—when ABW drops suddenly, high-bitrate packets continue to be sent, increasing queuing delays and worsening congestion.

Figure~\ref{fig:Motivation_filtering_methods_4} offers a closer examination of the impact of the bandwidth reduction on the frame delay. When ABW abruptly drops from 20 Mbps to 10 Mbps at the 40th second, queue occupancy ratio peaks at 40\% and the frame delay reaches nearly 600 ms, causing analytics results to lag behind the current frame and significantly degrading accuracy.
Table~\ref{tab:frame_delay_accuracy} further quantifies this effect: as the frame delay increases from 50 ms to 800 ms, person detection accuracy degradation worsen from 4.70\% to 23.51\%, and vehicle detection accuracy degradation worsen from 2.73\% to 27.38\%. These findings highlight the urgent need for an in-network solution that can dynamically adjust the amount of data transmitted to preserve accuracy.

\begin{figure}
    \centering
    \includegraphics[width=0.99\linewidth]{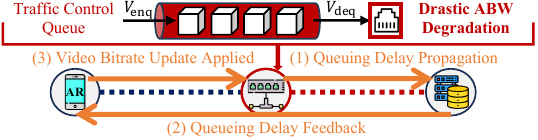}
    \vspace{0mm}
    \caption{Strawman approaches gradually adjust video bitrate to address ABW degradation at a last-mile router.}
    \label{fig:Motivation-0}
    \vspace{1em}
\end{figure}

\begin{table}[t!]
\caption{Object detection accuracy degradation for different levels of frame delay on Jacksonhole~\cite{jackson}  with YOLOX.}
\centering
\footnotesize
\begin{tabular}{|c|c|c|c|c|c|}
\hline
\textbf{\begin{tabular}[c]{@{}c@{}}Frame Delay\end{tabular}} & \textbf{50 ms} & \textbf{100 ms} & \textbf{200 ms} & \textbf{400 ms} & \textbf{800 ms} \\ \hline
\textbf{\begin{tabular}[c]{@{}c@{}}Person\end{tabular}}                                                           & 4.70\%$\downarrow$        & 5.90\%$\downarrow$       & 8.87\%$\downarrow$       & 14.53\%$\downarrow$       & 23.51\%$\downarrow$       \\ \hline
\textbf{\begin{tabular}[c]{@{}c@{}}Vehicle\end{tabular}}                                                          & 2.73\%$\downarrow$        & 3.91\%$\downarrow$       & 7.66\%$\downarrow$       & 16.26\%$\downarrow$       & 27.38\%$\downarrow$       \\ \hline
\end{tabular}
\vspace{2mm}
\vspace{0mm}
\label{tab:frame_delay_accuracy}
\end{table}

\begin{figure}[!t]
    \vspace{0mm}
    \centering
    %\subfigtopskip=2pt %设置子图与上面正文或别的内容的距离
    \subfigbottomskip=0mm %设置第二行子图与第一行子图的距离，即下面的头与上面的脚的距离
    \subfigcapskip=0mm %设置子图与子标题之间的距离
    % 第零行
    \subfigure[Detected vehicles in all regions.]{
		\includegraphics[width=0.46\linewidth]{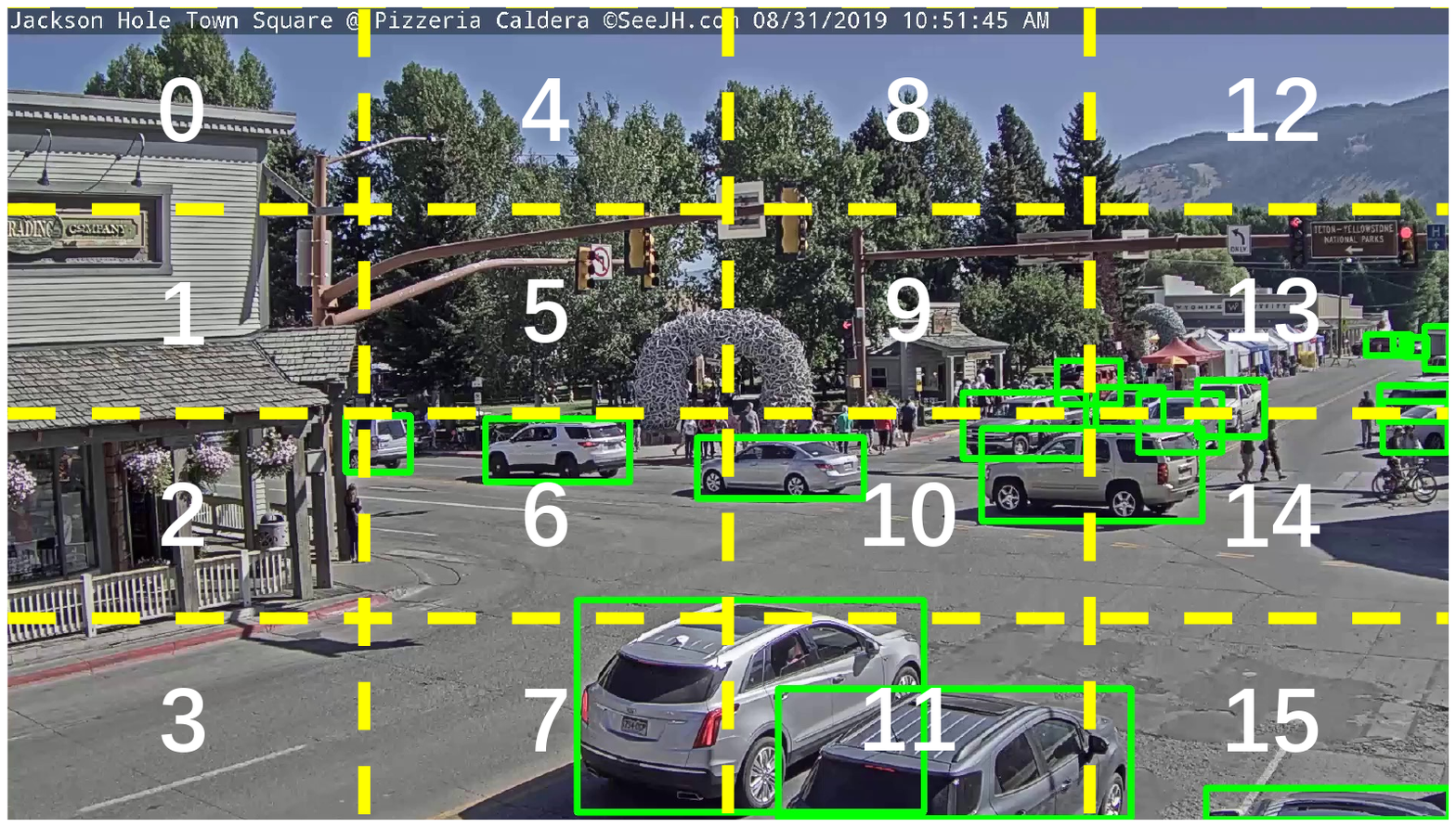}
		\label{fig:Motivation-2}
	}\hfill
    \subfigure[Average number of vehicles.]{
		\includegraphics[width=0.46\linewidth]{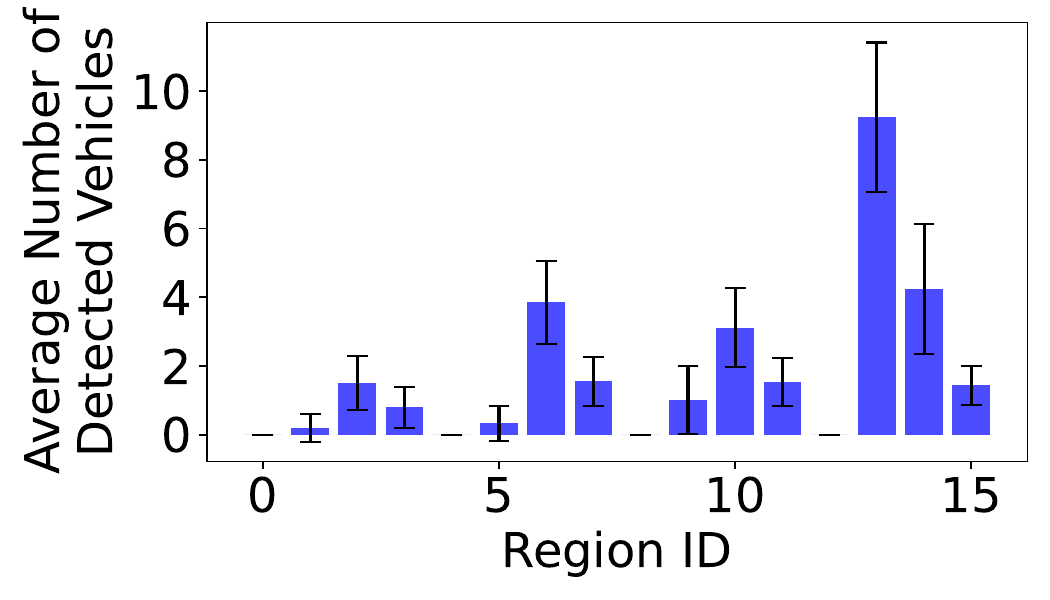}
		\label{fig:Motivation-3}
	}\\
    % 第一行
    \subfigure[Varying number of vehicles.]{
		\includegraphics[width=0.46\linewidth]{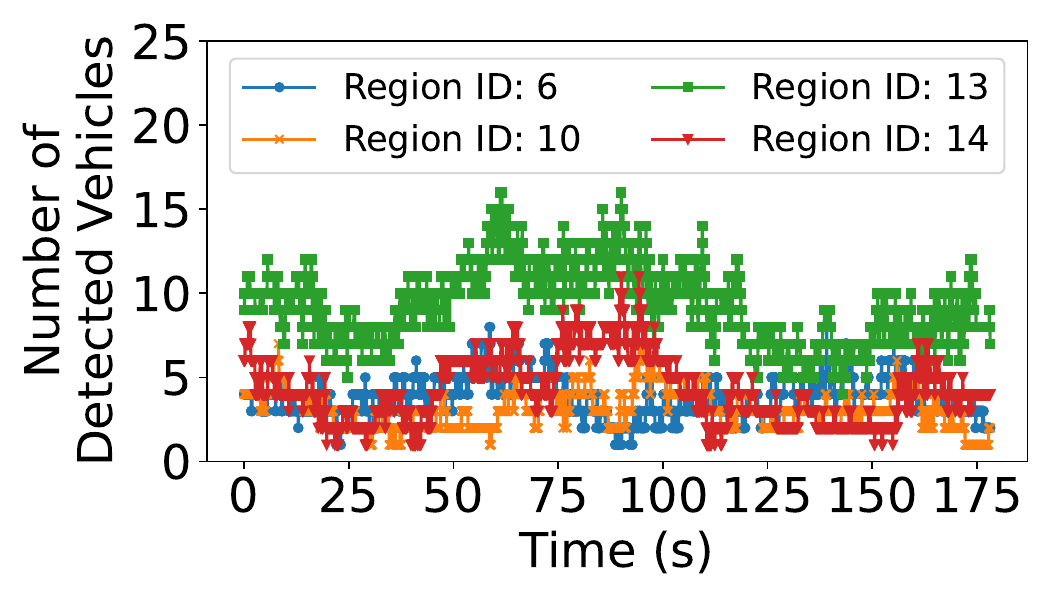}
		\label{fig:Motivation-4}
	}\hfill
    \subfigure[Numbers of vehicles and persons.]{
		\includegraphics[width=0.46\linewidth]{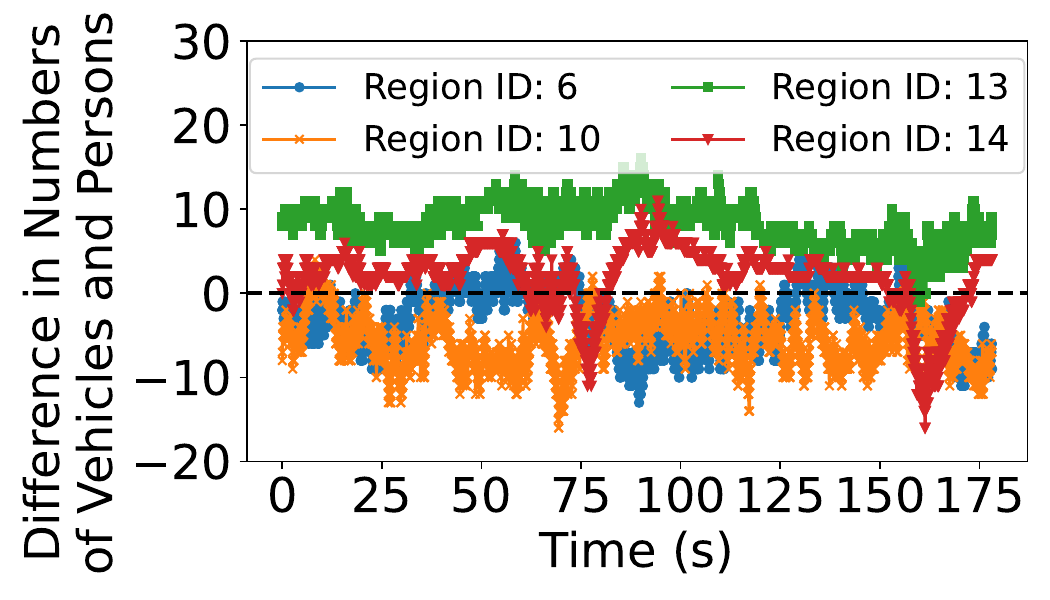}
		\label{fig:Motivation-4_2}
	}
    \vspace{0mm}
    \caption{Uneven distribution across frame-divided regions.}
    \label{fig:Motivation-2-3-4}
    \vspace{0mm}
\end{figure}

\MyPara{Opportunities.}
We observe an opportunity arising from the \emph{uneven distribution} across frame-divided regions.
Figure \ref{fig:Motivation-2} shows vehicle detection results from the Jacksonhole dataset \cite{jackson}, with the frame divided into 4$\times$4 regions. 
Figure \ref{fig:Motivation-3} reveals significant variations in average vehicle counts across regions, with regions 0, 4, 8, and 12 having no vehicles, and region 13 having the highest count. 
Figure \ref{fig:Motivation-4} shows temporal variations, with region 13 consistently having the most vehicles, while regions 6, 10, and 14 show fluctuating patterns. 
Figure \ref{fig:Motivation-4_2} illustrates the varying relative differences in detected vehicles and persons.  

In summary, Figure \ref{fig:Motivation-2-3-4} highlights that (1) regions contribute differently to video analytics quality, reflecting their varying importance, and (2) regional importance changes over time, requiring dynamic reassessment. These findings suggest opportunities to distinguish different blocks by assessing their individual importance to the analytics accuracy and drop those that contribute less when ABW degrades.

%% file: 4-design.tex
\mysection{Design
\label{sec:design}}
%\vspace{.25em}
%\mysubsection{Design Overview}

%\tool is a system designed to enhance real-time video analytics by prioritizing packet discarding based on each packet's impact on video analytics accuracy. 

\tool comprises three main components: block-aware encoding and packetization at the client side (\S\ref{subsec:encode-packet}), analytics-driven packet prioritization at the server side (\S\ref{subsec:priorityFeedback}), and adaptive packet discarding at last-mile routers (\S\ref{subsec:packetDiscarding}). Figure~\ref{fig:design-overview} illustrates \SYS's overall architecture.
%\hx{Rewrite this.}
%\shan{Done.}

\MyPara{Block-Aware Encoding and Packetization (Client Side).}
The video analytics client first encodes the real-time video frames with a \emph{block-based frame encoding} method (\S\ref{subsubsec:frame-encoding}). The method segments each frame into independent blocks, minimizing dependencies between them. Such segmentation is crucial as it ensures that frames can still be successfully decoded even if some packets are discarded. Moreover, it facilitates the clear distinction between blocks based on their significance to overall video analytics accuracy. The importance of each block is quantified by a \emph{cumulative size ratio (CSR)} (\S\ref{subsubsec:header-ext}) value, defined as the fraction of the frame’s size contributed by that block and all lower-priority blocks. The block priorities are derived from the analytics performed on the server side.

Blocks are then packetized using a \emph{block-aware RTP header extension} (\S\ref{subsubsec:header-ext}). The new RTP header extension is specifically designed to incorporate each block's metadata and its importance, ensuring that each block can be individually recognized and properly managed during network congestion. 

While block-based encoding effectively supports importance-based packet drops based on content, its encoding efficiency is 17.2\% lower (\ie sending more data) than traditional frame-based encoding across different bitrates. 
However, for real-time analytics where resilience to frame loss is critical \cite{mackay1997near, badr2017fec, rudow2023tambur, li2023reparo, cheng2024grace, kang2022error, sankisa2018video}, the added flexibility for resilience and targeted adaptation is far more important than mild increases in the amount of data transferred.

\begin{figure}[t]
    \centering
    \includegraphics[width=0.8\linewidth]{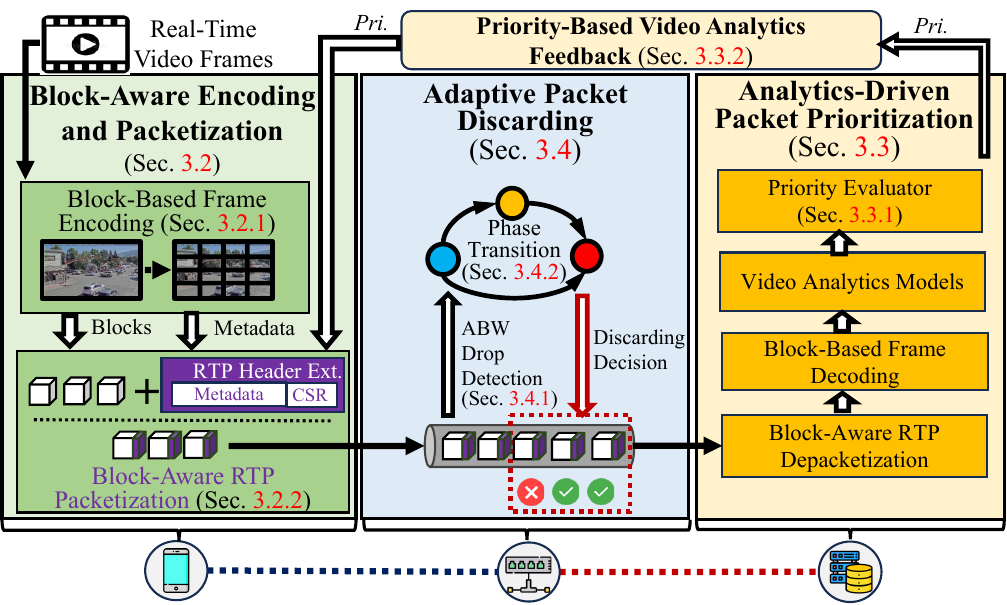}
    \caption{Overview and workflow of \SYS.}
    \vspace{0mm}
    \label{fig:design-overview}
\end{figure}

\MyPara{Analytics-Driven Packet Prioritization (Edge-Server Side).}
Upon receiving packets, the analytics server performs depacketization and decodes packets into blocks. Thanks to the resilient block-based encoding and our packet discarding strategy, even with some of the low-importance blocks dropped, the integrity and accuracy of video analytics are minimally impacted. 
After analytics, a priority evaluator re-accesses each block's priority using the RAI metric (\S\ref{subsubsec:accuracy_driven}). Any changes in the priority values are promptly communicated back to the client through an \emph{RTCP-based feedback loop} (\S\ref{subsubsec:rtcp_based}), ensuring that the client's packetization process remains aligned with the most current analytics-driven priorities.

\MyPara{Adaptive Packet Discarding (Edge Routers).} 
To enable adaptive packet discarding, edge router continuously tracks enqueue and dequeue packet rates to detect any reduction in available bandwidth (\S\ref{subsubsec:network-metric}). Upon detecting a reduction, our algorithm triggers a phase transition to proactively discard packets (\S\ref{subsubsec:state}). During this phase, a drop threshold is calculated adaptively based on current network conditions. The packet discarding manager then eliminates all packets whose \emph{cumulative size ratio (CSR)} (\S\ref{subsubsec:header-ext}) falls below the specified threshold, alleviating network congestion while preserving the transmission of high-priority blocks critical to analytics accuracy. This filtering only inspects the CSR information encoded in the packet headers, without needing to maintain any per-flow state on edge routers.

% We present an overview of our proposed system \SYS in Figure \ref{fig:Design-0}. 
% Corresponding to the three challenges mentioned earlier, it consists of three key components:
% \vspace{-0.5em}
% \begin{itemize}[itemsep=0.5em, leftmargin=1em]
% \item To minimize the negative impact of retransmission on the overall frame latency, we carefully devise a block-aware RTP header extension (Section \ref{subsec:encode-packet}) at the client.
% It effectively reduce dependencies between video packets by embedding crucial block-related information into them.\vspace{-0.75em}
% \item To simultaneously account for both frame structure and network conditions, we propose a state-transition-based packet discarding strategy (Section \ref{subsec:packetDiscarding}) at the bottleneck router, enabling efficient and precise packet discarding decisions.\vspace{-0.75em}
% \item Leveraging an application-specific metric of relative accuracy impact (RAI) for dynamically evaluating packet priorities, we propose a priority-based video analytics feedback mechanism (Section \ref{subsec:priorityFeedback}) on the server to support \SYS’s proactive packet discarding.\vspace{-0.5em}
% \end{itemize}
% \shan{I feel a little hard to understand the relationships between each module, and how each module contributes to our goals exactly.}

% \begin{figure}[t!]
%     \centering
%     \includegraphics[width=1\linewidth]{figure/Design-0.pdf}
%     \vspace{-4mm}
%     \caption{Overview of our proposed system \SYS.}
%     \vspace{0mm}
%     \label{fig:Design-0}
% \end{figure}

\mysubsection{Block-aware Encoding and Packetizing\label{subsec:encode-packet}}
\mysubsubsection{Block-Based Frame Encoding\label{subsubsec:frame-encoding}}
%$\newline$
We aim to discard packets at edge routers while ensuring the retained packets remain decodable and analyzable at the server without retransmission. 
The strong correlation \cite{liang2008analysis, cheng2024grace} among packets encoded by traditional methods \cite{vp8, vp9, H264} poses a challenge.
To address this, it is essential to reduce inter-packet dependencies without hurting encoding quality.

To balance encoding efficiency and adaptable packet dependencies, 
% To this end, 
we adopt a block-based frame encoding technique \cite{du2022accmpeg, shi2023adapyramid, guan2019pano, chen2024tilesr}. % to balance encoding quality and adaptable packet dependencies.
At the client, each frame is divided into multiple blocks (\eg N$\times$N \cite{du2022accmpeg, shi2023adapyramid, guan2019pano, chen2024tilesr}), which are independently encoded and packetized for transmission, as shown in Figure \ref{fig:Design-2}. Across frames, \tool follows the conventional I/P-frame encoding, ensuring intra-frame block independence with minimal impact on encoding efficiency.

The value of $N$ can be tuned to balance the trade-off between the impact of discarding packets and encoding efficiency\textemdash A larger $N$ (fine-grained) allows \tool to better preserve accuracy when dropping packets but increases encoding redundancies (\ie more data transmitted). Our results show that N=4 strikes a good balance; a detailed evaluation of the impact of $N$ can be found in Figure \ref{fig:Eva_parameter_bsize} in \S\ref{sec:evaluation}. 
Different from encoding the entire frame, the target bitrate $B_{target}$ needs to be allocated across multiple blocks:\vspace{0mm}
\begin{equation}
    \sum\nolimits_{blk \in \bm{b}}Alloc(blk)\leq B_{target}.\vspace{0mm}
\end{equation}
$Alloc$ represents the specific bitrate allocation method across multiple blocks $\bm{b}$, which can be adaptively and flexibly decided (e.g., \cite{du2022accmpeg, shi2023adapyramid, guan2019pano, chen2024tilesr}).
Each block $blk$ is encoded at the target bitrate $Alloc(blk)$ and then packetized into RTP packets.
This method preserves the inherent correlation among packets within each block, thus maintaining high video encoding efficiency. 
Additionally, the independence of packets across different blocks ensures that discarding packets from one block does not affect the decoding, depacketization, and further analysis of others.
This enables \SYS to flexibly discard packets at the last-mile bottleneck router.

%\shan{todo: remove and merge to next section}

\begin{figure}[t!]
    \centering
    \includegraphics[width=1\linewidth]{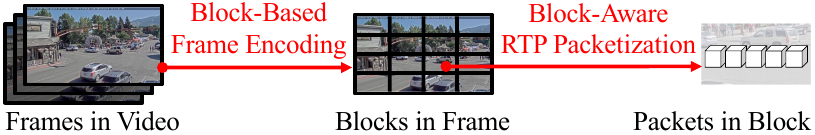}
    \caption{Frames are divided into blocks, encoded, packetized with block-specific metadata, and transmitted.}
    \vspace{0mm}
    \label{fig:Design-2}
\end{figure}

\mysubsubsection{Block-Aware RTP Header Extension\label{subsubsec:header-ext}}
%$\newline$
%\MyPara{Limitations of existing RTP header.}
%The existing RTP packet header structure cannot directly support the aforementioned block-based video packet decoupling, which is essential for \SYS's proactive block discarding.
%As illustrated in Figure \ref{fig:Design-1} (upper part), the RTP header includes the RTP version, payload type, sequence number, timestamp, and other elements.
%This information, along with the RTCP control protocol \cite{sarker2021rtp, novotny2008large}, ensures the reliable transmission of standard video media.
%However, the existing RTP packet header structure cannot directly support the aforementioned block-based video packet decoupling, which is essential for \SYS's proactive block discarding.
%This motivates us to further enhance the capabilities of the RTP header by designing the RTP header extension.
%Figure \ref{fig:Design-1} (upper part) illustrates the RTP header \cite{frederick2003rtp}, which includes the payload type, sequence number, and other elements. 
To enable frame-aware packet discarding at the last-mile router, it is essential that packets contain information that can assess their importance for discarding. 
Additionally, these packets must include metadata that allows the receiver to reassemble and analyze frames, even if some packets are discarded. A straightforward solution is to reuse the differentiated services (DS) field in the IPv4/IPv6 header \cite{nichols1998definition} to encode the block importance. 
%However, this header field only has 6 bits and is insufficient for supporting block-based video packet decoupling \hx{more specific?}. 
However, this header field only has 6 bits, limiting its ability to store granular block-level metadata (described later) needed for block-based frame packetization.
To solve this problem, we shift our focus to RTP, which supports a header extension mechanism, allowing custom information (such as application-specific data) to be added to packets.
This mechanism does not impact the transmission of video packets over the existing network and remains fully backward compatible \cite{sarker2021rtp}.
By utilizing the RTP header extension \cite{frederick2003rtp, schulzrinne2003rfc3550, gao2019rtp}, we carefully design our extension structure  (illustrated by the green part in Figure \ref{fig:Design-1}) with minimal extension length. This minimal length limits the number of blocks N to be $\leq$ 4. Although our empirical results (Figure~\ref{fig:Eva_parameter_bsize}) show that the 4$\times$4 encoding can already provide excellent accuracy preservation, users can opt for a larger N by extending the header length if needed.
%\rn{make sure to elevate this design choice to the intro!
%\yu{OK.}

Next, we present three important goals, and discuss how our design achieves them.

\MyPara{Frame reconstruction with block-to-frame metadata.}
    %The existing RTP header structure only considers basic information at the packet level and does not include relevant information for each block within a video frame.
The absence of block-related information makes it hard for the video analytics server to properly decode and reconstruct the original frame from the received packets for analysis.

We devise the indicators for block start (\textbf{S}) and end (\textbf{E}), which are embedded in each block’s packets at the client and are each allocated 1 bit. 
They signify whether the current packet marks the beginning or end of a block, thus facilitating the receiver in performing depacketization and decoding once a complete block is received.
Additionally, the server reconstructs frames by organizing each block according to its positional coordinates (\textbf{X}, \textbf{Y}), which are embedded in the packets at the client.
Specifically, each block is restored to its specific location ($x$, $y$) at the server as $x = {\textbf{X}} W/{|\{\textbf{X}\}|}$ and $y = {\textbf{Y}} H/{|\{\textbf{Y}\}|}$,
%\vspace{0mm}
%\begin{equation}
%    x = {\textbf{X}} W/{|\{\textbf{X}\}|},\quad y = {\textbf{Y}} H/{|\{\textbf{Y}\}|},\vspace{0mm}
%\end{equation}
where $W$ and $H$ represent the original frame width and height, respectively. 
Thus, $W/{|\{\textbf{X}\}|}$ and $H/{|\{\textbf{Y}\}|}$ denote the block width and height accordingly.
The above values of width and height are decided when the client and server establish their initial connection, and are used thereafter.
With each frame uniformly divided into 4$\times$4 blocks, 4 bits are allocated to represent each block's position.

\begin{figure}[t!]
    \centering
    \includegraphics[width=0.9\linewidth]{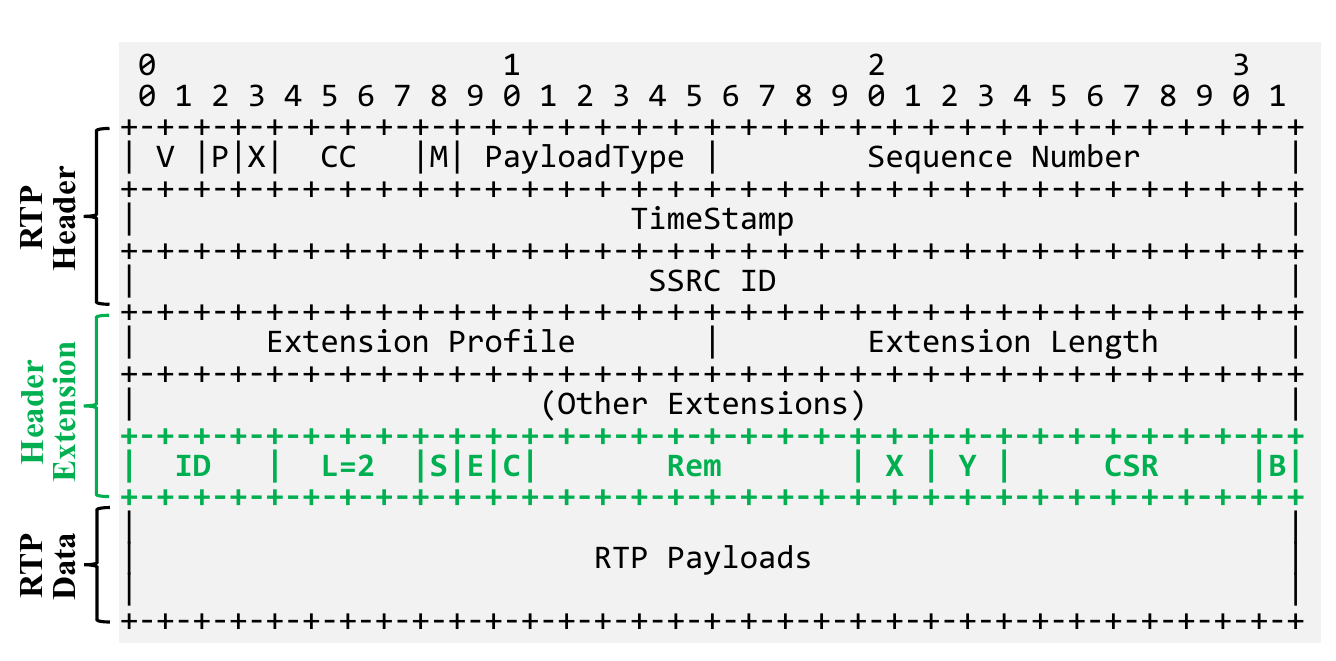}
    \vspace{2mm}
    \caption{The RTP header structure with block-aware RTP header extension for proactive packet discarding. }
    \vspace{0mm}
    \label{fig:Design-1}
\end{figure}

\MyPara{Block priority-based packet discarding.}
When a drastic available bandwidth degradation is detected at an edge router, \SYS should selectively discard some blocks.
However, the existing RTP header structure lacks information about the importance of each block, making it hard for \SYS to make decisions on block discarding.
%\del{Thus, the importance of each block should be stored in the header extension for \SYS to make decisions on block discarding.}
To solve this problem, we derive each block's priority $\mathit{pri}$ from the historical analysis results provided by the video analytics server as a proxy of its importance. Following the uniform division of each frame into 4$\times$4 blocks, we assign each block a unique priority value ranging from 0 to 15, with a larger value indicating a higher priority.

In an edge router, when ABW degradation is detected, \SYS must decide on a priority threshold below which packets should be discarded. To this end, we introduce \emph{cumulative size ratio (CSR)}, a metric that quantifies the fraction of the frame data that would be eliminated by discarding all blocks whose priority values are $\leq$ a given level $\mathit{pri}$: %\del{(\ie, the higher a priority value, the less important a block)}: %
\begin{equation}
CSR(pri) = 
\frac{\sum_{0 \leq q \leq pri} S(q)}{\sum_{0 \leq q \leq 15} S(q)} .\vspace{-1mm}
\end{equation}
where $S(q)$ denotes the total size of all packets at priority $q$. 
% \begin{equation}
% \label{eq:csr}
% CSR(pri) \;=\;
% \frac{\sum_{p \in \mathcal{F},\;\mathrm{Pri}(p) \ge pri} \mathrm{size}(p)}
%      {\sum_{p \in \mathcal{F}} \mathrm{size}(p)} ,
% \end{equation}
% where $\mathcal{F}$ is the set of packets in the frame, $\mathrm{Pri}(p)$ is the priority of packet $p$, and $\mathrm{size}(p)$ its size. The numerator sums the sizes of packets with $\mathrm{Pri}(p)\!\ge\! pri$, and the denominator is the total frame size.
Intuitively, $CSR(pri)$ captures the potential bandwidth savings if blocks at priority $pri$ and lower are dropped.
%A higher $CSR$ value corresponds to more important packets, since retaining them preserves a larger fraction of the frame.

To make CSR available to the router, each packet is encoded at the client with the CSR value corresponding to its block priority in the RTP header extension. We reserve 7 bits for CSR, providing 128 discrete levels over the range [0,1]. Each frame's blocks are then encoded, packetized, and transmitted in descending order of their CSR values, ensuring that more important blocks arrive first.
During the ABW degradation at an edge router, \SYS simply compares the embedded $\mathit{CSR}$ in each packet against the observed ABW deficit ratio (computed at the router). Packets with $\mathit{CSR}$ lower than the deficit ratio are proactively discarded. The detailed packet discarding process is described in \S\ref{subsec:packetDiscarding}.

%\del{
%Consequently, we reserve 4 bits in the block-aware RTP header extension for block priorities.
%At the client, each frame's blocks are encoded, packetized, and transmitted over the network in descending order of priority.
%To improve real-time video analytics performance, \SYS utilizes \textbf{Pri} to proactively discard packets at the last-mile router when the priority of packet $pkt$ satisfies $pkt.\textnormal{\textbf{Pri}} \geq Dis$，
%where $Dis$ is the packet discarding threshold. The detailed process for packet discarding is provided in \S\ref{subsec:packetDiscarding}.}

\MyPara{Avoiding retransmissions for discarded blocks.}
% \shan{Todo: move to packet discarding section}
From the perspective of the server, \SYS's proactive packet discarding at an edge router results in packet loss. 
Without additional information stored in the RTP header extension, the server finds it difficult to distinguish between packet loss caused by \SYS and other unexpected cases (e.g., AQM mechanisms \cite{ryu2004advances, hollot2002analysis}, channel interferences \cite{kihero2021wireless, villegas2007effect}). 
Consequently, the server faces challenges in deciding whether to request retransmission and how to request it properly.

Specifically, when the server receives a video packet $pkt$ with a sequence number $pkt.\textnormal{\textbf{N}}$, it updates the sequence number of the next expected packet to $pkt.\textnormal{\textbf{N}} + 1$. This is guaranteed because the RTP stack maintains a reordering buffer and delivers packets to video analytics server in the order of sequence number.
If the sequence number of the actual packet received exceeds $pkt.\textnormal{\textbf{N}} + 1$, the video analytics server interprets this as packet loss in the network, triggering a retransmission request. 
To prevent unnecessary retransmissions under \SYS's proactive discarding, we introduce a one-bit boundary flag $B$. When $B$ is set for a packet, it marks the proactive dropping boundary: all subsequent packets in the same frame have been discarded. Upon receiving this boundary packet with sequence number $pkt.\textnormal{\textbf{N}}$, the server adjusts the next expected sequence number to $pkt.\textnormal{\textbf{N}} + pkt.\textnormal{\textbf{Rem}}$, where $pkt.\textnormal{\textbf{Rem}}$ is embedded by the client in the RTP header to indicate the number of remaining packets in the current frame. 
%\hx{Explain the reason why packets can be received in the same order as they were sent}

%\del{\textbf{Dis} and \textbf{Rem} to represent the discarded block ranges.
%\textbf{Dis} embedded in the packets at a last-mile router indicates the starting priority of the discarded blocks. 
%\textbf{Rem} embedded in the packets at the client represents the number of packets remaining in the current frame from this packet onward.
 %   Concretely, when the currently received packet $pkt$ satisfies: \vspace{-2mm}
%\begin{equation} 
%pkt.\textnormal{\textbf{Pri}} == pkt.\textnormal{\textbf{Dis}} - %1\quad\textnormal{and}\quad pkt.\textnormal{\textbf{E}} == 1, \vspace{-1mm}
%\end{equation} 
%the next expected packet's sequence number is updated to:\vspace{-1mm}
%\begin{equation} 
%pkt.\textnormal{\textbf{N}} + pkt.\textnormal{\textbf{Rem}}. \vspace{-2mm}
%\end{equation}
%Notably, the packets with sequence numbers between $pkt.\textnormal{\textbf{N}}$ %and $pkt.\textnormal{\textbf{N}} + pkt.\textnormal{\textbf{Rem}}$ are discarded %by \SYS at the router.
%Following \textbf{Pri}, \textbf{Dis} is allocated 4 bits.}

As shown in \apx{Appendix \ref{appendix:estimated_number_of_RTP_packets}}{the extended version~\cite{lizardExtended}}, the estimated number of packets in a frame for common video bitrates and frame rates is less than 500, motivating us to allocate 9 ($>\log_2 500$) bits for \textbf{Rem}.
Additionally, \textbf{C}, a 1-bit field, can be marked by \SYS at the edge router when a sudden and sharp drop in available bandwidth is detected, triggering proactive packet discarding.

%In summary, the structure of our designed block-aware RTP header extension is illustrated in Figure \ref{fig:Design-1} (middle part), with additional details provided in Appendix \ref{appendix:summary_for_rtp_extension}.
 
%also summarizes the embedder and parser for each field with bit allocation in block-aware RTP header extension.
%Corresponding to the limitations \textit{\textbf{L$_1$}}, \textit{\textbf{L$_2$}} and \textit{\textbf{L$_3$}} that need to be overcome,
%the fields in the block-aware RTP header extension are primarily divided into three categories: 

%\begin{algorithm}[t]
%    \small
    %\caption{RAI-Driven Block Priority Feedback}
    %\label{Alg:Design-2}
    %\KwIn{ 
      % $BlockRAI \leftarrow \{\}$;
     %  $BlockPriLast \leftarrow \{\}$
    %}
    %\While {\textnormal{Enqueue new frame $frm$}}{
      %  $result \leftarrow VideoAnalytics(frm)$\;
     %   \tcp{Block Prioritization}
        %\textcolor{blue}{
    %    \textcolor{black}{
   %     \For {\textnormal{Each $blk$ in frame $frm$}}{
  %          $<block, RAI> \leftarrow CalculateRAI(blk, result)$\;
 %           $BlockRAI.Add(<block, RAI>)$\;
 %       }
 %       $BlockPri \leftarrow Sort(BlockRAI)$\;
 %       }

%        \tcp{Priority Feedback}
        %\textcolor{orange}{
%        \textcolor{black}{
%        \If {\textnormal{$BlockPri != BlockPriLast$ and $FeedbackInterval() > \lambda$}}{
%            $SendFeedback(BlockPri)$\;
%            $BlockPriLast \leftarrow BlockPri$\;
%        }
%        }
%        $ReturnResult(result)$\;
%    }
%\end{algorithm}

\begin{figure}[!t]
        \vspace{-4mm}
	\centering
        \subfigcapskip=0mm %设置子图与子标题之间的距离
	\subfigure[Arbitrarily selected blocks.]{
		\includegraphics[width=0.43\linewidth]{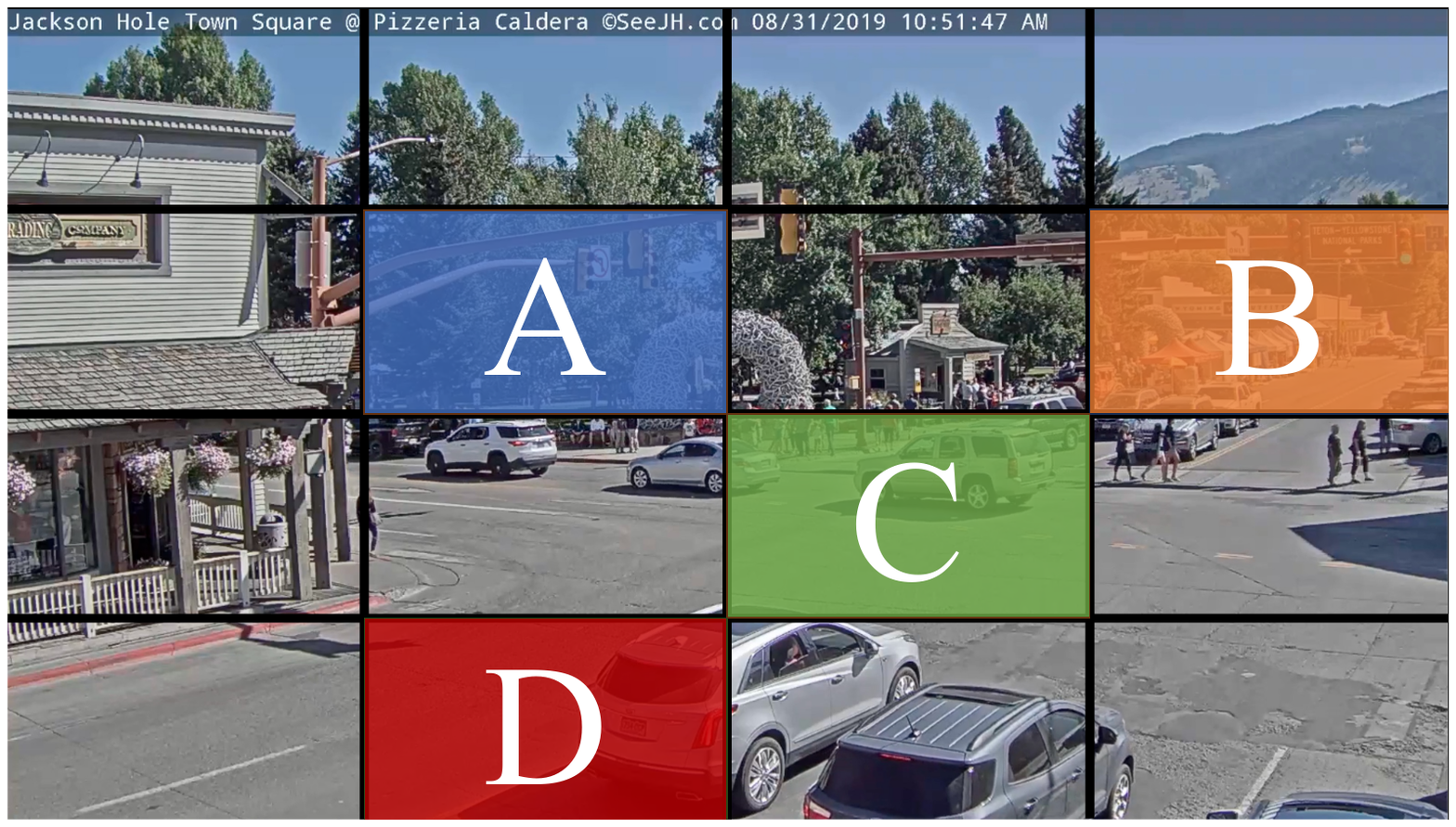}
            \label{fig:Design-9}
	}\hfill
	\subfigure[Varying block-based RAI.]{
		\includegraphics[width=0.43\linewidth]{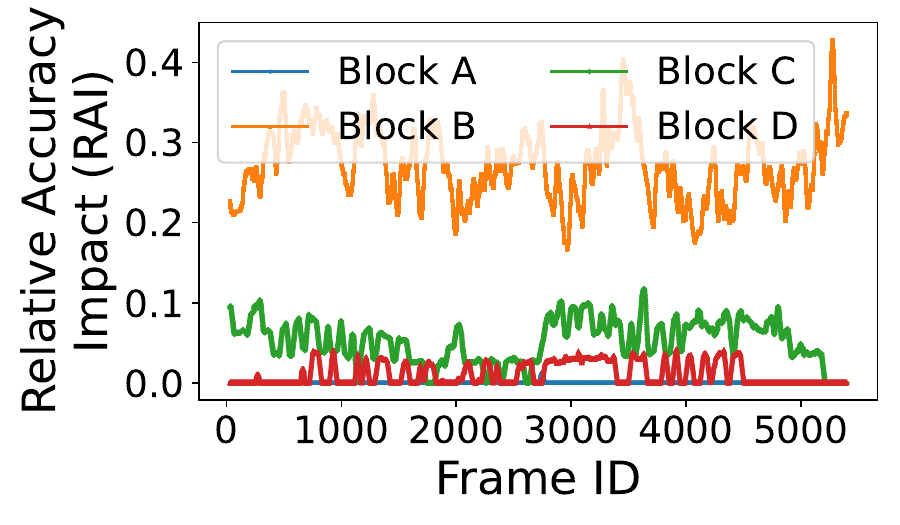}
            \label{fig:Design-10}
	}
	\caption{Varying RAIs for arbitrarily selected blocks.}
    \label{fig:Design-9_10}
\end{figure}

\mysubsection{Analytics-Driven Packet Prioritization
\label{subsec:priorityFeedback}}
%In the previous subsections, it was noted that the client embeds the importance of each block in a block-based RTP extension header before transmitting the RTP packets to the network, with the importance value sent by the server as feedback. To facilitate this, 
\SYS calculates the priority of each block on the server side, as the video analytics server has access to the complete analysis results. Subsequently, the block priority information is sent back to the client, which embeds it in the RTP extension headers before transmitting packets to the network. However, this approach raises two critical issues that must be addressed:
(1) How does \SYS define and calculate the priority of each block on the video analytics server? 
and (2) How does \SYS send the block priority information, computed on the video analytics server, back to the client?

%\vspace{-0.5em}
%\begin{itemize}[itemsep=0.5em, leftmargin=1.5em]
%    \item[1.] How does \SYS define and calculate the priority of each block on the video analytics server?
%    \vspace{-0.75em}
%    \item[2.] How does \SYS send the block priority information, computed on the video analytics server, back to the client?\vspace{-0.5em}
%\end{itemize} 

Before incorporating block priority feedback, the video analytics server processes each new incoming frame by feeding it into a deep neural network model (\eg YOLOX~\cite{ge2021yolox}) for video analytics. 
The analysis results for the current frame are then returned to the client.
\del{Based on this, we further propose an accuracy-driven block prioritization method and then devise an efficient RTCP packet structure for block priority feedback to address the two issues mentioned above.}

\begin{figure}[t!]
    \centering
    \vspace{2mm}
    \includegraphics[width=0.95\linewidth]{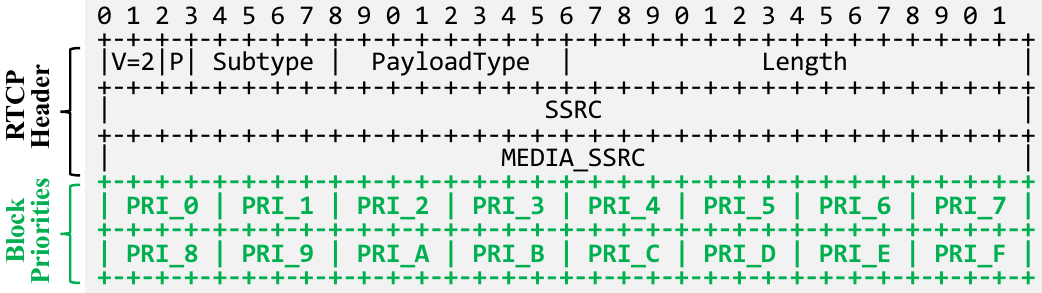}
    \caption{RTCP with accuracy-driven block priorities.}
    \vspace{0mm}
    \label{fig:Design-8}
\end{figure}

\mysubsubsection{Accuracy-Driven Block Prioritization\label{subsubsec:accuracy_driven}}
%$\newline$
To define the priority of each block, we consider its contribution to the overall accuracy of video analytics for the entire frame. 
To quantify this contribution, we propose a metric of block-based relative accuracy impact (RAI).
It reflects the improvement in accuracy gained by retaining a block compared to removing it within the current frame.
We represent the relative accuracy impact $\Delta A_{blk}$ for each block $blk$ as:\vspace{-2mm}
\begin{equation}
    \label{Equ:RelativeAccuracyImpact}
    \Delta A_{blk} = 
    ({A_{\bm{b}} - A_{\bm{b} \setminus blk}})/{A_{\bm{b}}},\vspace{-2mm}
\end{equation}
where $A_{\bm{b}}$ denotes the video analytics accuracy of the current frame with all blocks, and $A_{\bm{b} \setminus blk}$ denotes the accuracy after removing block $blk$. \SYS uses RAI to estimate each block's recent importance. For example, when RAI is defined by the number of detected vehicles, the parking-area block B in Figure~\ref{fig:Design-9_10} has the highest RAI, the landscape block A the lowest, while the two intersection blocks alternate in priority as vehicles move through them. \rev{RAI naturally extends to other tasks by replacing $A$ with task-specific metrics such as mIoU for segmentation, OKS for pose estimation, or rank-$k$ matching for re-identification, without changing the header extension or router logic. For tasks with weak spatial locality, RAI becomes flatter and \SYS gradually approaches content-agnostic dropping.}

% Upon the arrival and analysis of each new frame, the relative accuracy impact of each block in the current frame is calculated using Equation (\ref{Equ:RelativeAccuracyImpact}). 
% The RAI results for all blocks are then added to $BlockRAI$ and sorted into $BlockPri$ for priority feedback.

% \MyPara{RAI examples.}
% \SYS can leverage RAI to evaluate each block's priority over the recent period.
% For example, using the number of detected vehicles (which is customizable based on the specific task) as a metric for RAI, we can easily see that \textit{Block B} in Figure~\ref{fig:Design-9_10} has the greatest impact on accuracy as it includes a parking area with many vehicles. 
% In contrast, \textit{Block A}, containing a natural landscape area, has the least impact on the accuracy.
% The RAIs for \textit{Block C} and \textit{D}, which include moving vehicles at the intersection, dominate alternately.

\mysubsubsection{Priority-based Feedback\label{subsubsec:rtcp_based}}
%\textbf{RTCP APP packets.} 
%The RTCP protocol \cite{sarker2021rtp, novotny2008large} complements RTP by monitoring and maintaining the quality of streaming media. RTCP provides critical feedback on metrics such as packet loss and jitter, enabling real-time adjustments to enhance the user experience. Among the various RTCP packet types, application-specific (APP) packets are particularly notable for their ability to transmit customized information tailored to the needs of specific applications. These APP packets facilitate the integration of specialized feedback and control mechanisms, which are crucial for optimizing performance in scenarios such as video conferencing and online gaming.
%\shan{We may consider move this part into background, and briefly mention RTCP APP in the next part.}
%\yu{OK!}

%\textbf{RTCP-based priority feedback.}
%As we know, the RTCP protocol \cite{sarker2021rtp, novotny2008large} complements RTP by monitoring and maintaining the quality of streaming media. 
%RTCP provides critical feedback on metrics such as packet loss and jitter, enabling real-time adjustments to enhance the user experience. 
%$\newline$
The second issue is how to transmit the feedback for different regions from the edge server back to the client, so that it efficiently encode such information into new packets.
To resolve this issue, we leverage RTCP \cite{sarker2021rtp, novotny2008large}, which works alongside RTP to collect statistics (such as packet loss, jitter, and delay) and provide feedback to the client. 
%Fortunately, the Real-Time Control Protocol (RTCP) \cite{sarker2021rtp, novotny2008large} works alongside RTP to monitor data transmission quality. 
%It collects statistics such as packet loss, jitter, and delay, and provides feedback to the sender. 
%This feedback mechanism in RTCP can be leveraged to support efficient packet discarding at the bottleneck router when addressing drastic ABW degradation.
Among the various RTCP packet types, application-specific (APP) packets \cite{montagud2012enhanced} are especially significant as they enable customized feedback tailored to the specific application.

As illustrated in Figure \ref{fig:Design-8}, we extend the RTCP APP packet structure to incorporate priority information. 
Consistent with the bit allocation for priorities defined in the RTP header extension (\S\ref{subsec:encode-packet}), we reserve 4 bits for each block to represent its priority in the current frame. 
This adds only 8 bytes per feedback message, negligible compared to a typical 40–60 byte RTCP packet~\cite{schulzrinne2003rfc3550} and well within the 5\% RTCP bandwidth share recommended by RFC 3550, even with per-frame feedback.
When the current priority order $BlockPri$ differs from the previous $BlockOrderLast$, each block's priority is embedded into the structure. 
%Additionally, we introduce a feedback coefficient $\lambda$ to flexibly adjust the frequency of the priority-based feedback as $FeedbackInterval > \lambda$.\shan{Shall we remove $lambda$? we didn't evaluate and explain well how to set $lambda$}
%\begin{equation}
%    FeedbackInterval > \lambda.
This design enables real-time feedback, ensuring that critical video content is prioritized, thus facilitating proactive packet discarding at edge routers.
% \shan{It would be helpful if we could include, the best choice of lambda under different cases. I thought we should send the feedback immediately, since the video content on the server is already stale.}
%\shan{Maybe consider adding more explanation on how our feedback handles the moving camera cases, and how we handle the delayed or outdated feedback issue.} 
Figure \ref{fig:Eva_overall_MobileCamera_2} depicts the variation in feedback intervals across different video datasets. \SYS adaptively increases the feedback frequency, particularly for moving camera videos, to ensure priorities for packet discarding remain up-to-date.
%\S \ref{sec:evaluation} (\textcolor{blue}{Todo...}).
%\shan{ we can briefly talk about how the feedback mechanism deals with delays, moving camera, the sensitivity, besides a pointer to evaluation section. }

\MyPara{\rev{Bounding priority staleness.}}

\rev{Priority staleness is naturally bounded because packet discarding is transient and priorities are recomputed from every received frame. To further guard against sustained congestion, \SYS uses lightweight client-side priority aging and forced exploration, periodically promoting long-suppressed blocks so that newly important regions are eventually observed.}

%\begin{figure}[t!]
%    \centering
%    \includegraphics[width=1\linewidth]{figure/Design-5.pdf}
%    \vspace{-4mm}
%    \caption{Tracking packet sizes and priorities, and updating packet enqueue and dequeue rates over time.}
%    \vspace{-4mm}
%    \label{fig:Design-5}
%\end{figure}

\mysubsection{Adaptive Packet Discarding
\label{subsec:packetDiscarding}}
%\shan{We can explain that the changes can be implemented easily on the Linux-based edge routers first.}
As most last-mile routers in edge networks are Linux-based \cite{Anurag, vieira2020fast} and support flexible traffic management via software, we leverage this capability to implement adaptive packet discarding at an edge router with eBPF \cite{vieira2020fast}. However, since the packet dropping algorithm that runs on an edge server only performs a few checks (see Algorithm~\ref{Alg:Design-1}), it can be easily re-implemented with any vendor SDKs or ASIC APIs.

\mysubsubsection{ABW Drop Detection\label{subsubsec:network-metric}}
\MyPara{Updating enqueue \& dequeue rates.}
To enable \SYS's proactive packet discarding, the edge router estimates the enqueue and dequeue rates as indicators of the ABW conditions. 
%Figure \ref{fig:Design-5} illustrates how the ABW monitor records and updates these rates over time by tracking packet sizes. 
Upon the arrival of each packet, its size is recorded; however, the instantaneous enqueue and dequeue rates are not calculated immediately. 
Instead, the calculation is deferred until the time interval since the last rate update reaches a duration of $T_{win}$.
Taking both accuracy and timeliness into account, we judiciously set $T_{win}$ to 100 ms (see further analysis in \apx{Appendix \ref{appendix:analysis_for_t_win}}{\cite{lizardExtended}}).
As a result, at current time $t_0$, we calculate the packet enqueue rate as $v_{enq}(t_0) =$\vspace{-0.5mm}
%As shown in Figure \ref{fig:Design-5}, at current time $t_0$, we calculate the packet enqueue rate as
\begin{equation}
    \label{Equ:v_enq}
     %\frac{pktSizeAccum(t_0) - pktSizeAccum(t_1)}{t_0 - t_1},
     (pktSizeAccum(t_0) - pktSizeAccum(t_1))/(t_0 - t_1),\vspace{-0.5mm}
\end{equation}
where $t_0 - t_1$ slightly exceeds $T_{win}$, and $pktSizeAccum(t)$ represents the cumulative size of enqueued packets from the start up to time $t$.
The packet dequeue rate $v_{deq}(t_0)$ is calculated similarly to Equation (\ref{Equ:v_enq}) and is omitted here.

%\begin{figure}[t!]
%    \centering
%    \includegraphics[width=1\linewidth]{figure/Design-4.pdf}
%    \vspace{-4mm}
%    \caption{State transitions between \textbf{INIT}, \textbf{PRE} and \textbf{PD} for packet discarding at the bottleneck router.}
%    \vspace{-2mm}
%    \label{fig:Design-4}
%\end{figure}
  
\MyPara{ABW drop detection.} 
The ABW monitor actively analyzes the enqueue and dequeue rates to detect any ABW drop conditions, denoted as $DetectAbwDrop$. 
Specifically, when a new packet is enqueued, $DetectAbwDrop$ checks whether the following two conditions are simultaneously met:\vspace{-0.5mm}
%\begin{equation}
%    \label{Equ:DetectAbwDrop_1}
%    v_{deq}(t_0) < \gamma_1 v_{enq}(t_0),
%\end{equation}
%\begin{equation}
%    \label{Equ:DetectAbwDrop_2}
%    v_{deq}(t_0) < \gamma_2 v_{deq}(t_1),
%\end{equation}
\begin{equation} 
\label{Equ:DetectAbwDrop_1}
v_{deq}(t_0) < \gamma_1 v_{enq}(t_0)\quad\textnormal{and}\quad v_{deq}(t_0) < \gamma_2 v_{deq}(t_1), \vspace{-0.5mm}
\end{equation} 
where $0 < \gamma_1, \gamma_2 < 1$ represent the detection sensitivity coefficients, and they can be adjusted (e.g., $\gamma_1 = 0.7, \gamma_2 = 0.7$) to promptly detect different degrees of ABW drops.
The first condition in Equation (\ref{Equ:DetectAbwDrop_1}) reflects the difference between dequeue and enqueue rates at current time $t_0$, while the second condition captures the decline in the dequeue rate from the previous time $t_1$ to the current $t_0$.
Furthermore, we prove that satisfying both of the two conditions is essential for accurately detecting an ABW drop in \apx{Appendix \ref{appendix:proof_of_proposition}}{\cite{lizardExtended}}. 
%Furthermore, we prove that verifying the simultaneous satisfaction of the two conditions is essential for accurately detecting an ABW drop in Proposition \ref{thrm:DetectAbwDrop} in Appendix \ref{appendix:proof_of_proposition}. 

\MyPara{Detecting ABW returning to normal.} After an ABW drop is detected, the ABW monitor keeps monitoring to determine whether the ABW has returned to a steady state, denoted as  $CheckAbwSteady$.
The most straightforward idea for is that as ABW stabilizes, the previously much higher enqueue rate $v_{enq}$ will gradually get close to the dequeue rate $v_{deq}$, i.e.,\vspace{-1mm}
\begin{equation}
    \label{Equ:checkAbwSteady1}
    {|v_{enq} - v_{deq}|}/{v_{deq}} \leq \epsilon,\vspace{-1mm}
\end{equation}
where $\epsilon$ (e.g., $\epsilon = 0.1$) represents a small quantity, indicating a minimal difference between the enqueue and dequeue rates.
However, due to proactive packet discarding, the dequeue rate reflects video data rate after packet discarding, which is always significantly lower than the enqueue rate. 
Therefore, Equation (\ref{Equ:checkAbwSteady1}) cannot be directly applied to detect when ABW has returned to a steady state. 
Instead, we compare the upper bound of the dequeue rate $\widehat{v}_{deq}$ with the enqueue rate for the $CheckAbwSteady$ process.
In other words, the dequeue rate $v_{deq}$ in Equation (\ref{Equ:checkAbwSteady1}) should be replaced by $\widehat{v}_{deq}$.
%expressed as:
%\begin{equation}
%    \label{Equ:checkAbwSteady1}
%    {|v_{enq} - \widehat{v}_{deq}|}/{\widehat{v}_{deq}} \leq \epsilon.
%\end{equation}
Since we cannot easily obtain the upper bound of the dequeue rate $\widehat{v}_{deq}$, we set it to the last recorded dequeue rate before the ABW drop is detected by $DetectAbwDrop$.

\begin{algorithm}[t]
    \footnotesize
    \caption{Phase Transition for Packet Discarding}
    \label{Alg:Design-1}
    \KwIn{
        $phase \leftarrow \textnormal{\textbf{INIT}}$; 
        $v_{enq} \leftarrow 0$;
        $v_{deq} \leftarrow 0$;
        % $setB \leftarrow \textnormal{False}$;
    }
    \While {\textnormal{Enqueue new packet $pkt$}}{
        \If{$phase == \textnormal{\textbf{INIT}} \ \&\& \ DetectAbwDrop()$}{
            $phase \leftarrow \textnormal{\textbf{PRE}}$; 
            $\tau \leftarrow \theta \cdot \max(0,\,1 - v_{deq}/v_{enq})$;
        }
        \ElseIf{$phase == \textnormal{\textbf{PRE}}$}{
            \If{$Check1stPktInFrame(pkt.\textnormal{\textbf{CSR}}, pkt.\textnormal{\textbf{S}})$}{
                $phase \leftarrow \textnormal{\textbf{PD}}$; 
                $pkt.\textnormal{\textbf{C}} \leftarrow \textnormal{True}$;
            }
        }
        \ElseIf{$phase == \textnormal{\textbf{PD}}$}{
            \If{$CheckAbwSteady()$}{
                $phase \leftarrow \textnormal{\textbf{INIT}}$;
                % $setB \leftarrow \textnormal{False}$;
            }
            \ElseIf{$pkt.\textnormal{\textbf{CSR}} \geq \tau$}{
                $pkt.\textnormal{\textbf{C}} \leftarrow True$; $setB \leftarrow \textnormal{False}$;
                }
            \ElseIf{$Check1stPktNextBlk(pkt.\textnormal{\textbf{S}}, setB)$}{
                    $pkt.\textnormal{\textbf{C}} \leftarrow \textnormal{True}$;
                    $pkt.\textnormal{\textbf{B}} \leftarrow \textnormal{True}$;
                    $setB \leftarrow \textnormal{True}$; 
                }
                \Else{
                    DiscardPacket($pkt$); Continue;
                }
            }
        }
        Dequeue packet $pkt$; \textit{Update}$(v_{enq}, v_{deq})$;
\end{algorithm}

\subsubsection{Phased Packet Discarding\label{subsubsec:state}}
\MyPara{Phase transitions.}
In our \SYS system, we define three phases for packet discarding at an edge router: \textbf{INIT}, \textbf{PRE}, and \textbf{PD}.
%as in Figure \ref{fig:Design-4}.
The phase transitions for packet discarding are outlined in Algorithm \ref{Alg:Design-1}, where the initial phase is set to \textbf{INIT}. 
% When a new packet enters the queue, its size and priority are tracked, as indicated in line 2. 
Subsequently, as described in lines 2-15, the last-mile router responds to the current phase, with potential phase transitions occurring under specific conditions. 
Finally, the packet is dequeued, and the enqueue and dequeue rates are updated according to Equation (\ref{Equ:v_enq}), as shown in line 16. 
Based on Algorithm \ref{Alg:Design-1}, we detail all phases and their transitions as follows:

(1) \textbf{INIT} for ABW drop detection (\textit{lines 2-3}):
In the \textbf{INIT} phase, we continuously monitor for a $\mathit{DetectAbwDrop}$ signal from the ABW monitor. Once the signal is received, indicating that a drop in ABW is detected, Algorithm \ref{Alg:Design-1}
computes a drop threshold $\tau$ based on the current bandwidth deficit ratio:
\begin{equation}
   \label{equ:delayMitigation}
    \tau \leftarrow \theta \cdot \max(0,\,1 - v_{deq}/v_{enq})
\end{equation}
where the bandwidth deficit is calculated with $(1 - v_{deq}/v_{enq})$, indicating the fraction of incoming traffic that exceeds the capability at the edge router. $\theta$ is a configurable delay mitigation coefficient (\eg $\theta = 1$ by default) that controls the aggressiveness of discarding to alleviate the queuing delay at the router.

%\del{calculates the discarding priority $Dis$, which will guide packet discarding in the later \textbf{PD} phase.
%In line 5 of Algorithm \ref{Alg:Design-1}, $DecideDisPri$ is employed to determine the discarding priority $Dis$ based on the historical records of packet sizes and priorities as:\vspace{-1mm}
%\begin{equation}
%    Dis \leftarrow \underset{0 \leq pri < |\textnormal{\textbf{Pri}}|}{\arg\min} \frac{\sum_{pkt.\textnormal{\textbf{Pri}} \geq pri} pkt.size}{\sum_{pkt.\textnormal{\textbf{Pri}} \geq 0} pkt.size}\vspace{-1mm}
%\end{equation}\vspace{-1mm}
%\begin{equation}
 %   \label{equ:delayMitigation}
 %   s.t. \quad 
 %   \frac{\sum_{pkt.\textnormal{\textbf{Pri}} \geq pri} pkt.size}%{\sum_{pkt.\textnormal{\textbf{Pri}} \geq 0} pkt.size}  \geq \theta (1 - %\frac{v_{deq}}{v_{enq}}).\vspace{0mm}
%\end{equation}
%The discarding priority $\mathit{Dis}$ is set as a threshold, where packets with priority values larger than $\mathit{Dis}$ are discarded to alleviate queuing delays caused by the imbalance between enqueue and dequeue rates.
%Besides, $\theta$ denotes the delay mitigation coefficient and can be preset (\eg $\theta = 1$) to alleviate queuing delay at varying degrees at the edge router.}
% \yq{This algorithm seems also assuming some per-flow state to determine Dis for each client.}

(2) \textbf{PRE} for preparing packet discarding (\textit{lines 4-6}): 
Upon detecting an ABW drop, the \textbf{INIT} phase transitions to the \textbf{PRE} phase. 
The \textbf{PRE} phase acts as an intermediary between the \textbf{INIT} and \textbf{PD} phases, preventing immediate proactive packet discarding after the ABW drop is detected. 
This transitional phase is necessary for the following two reasons:
\vspace{-0em}
\squishlist
    \item As shown in Figure \ref{fig:Design-7},  upon the detection of an ABW drop, the current packet whose $\mathit{pkt}.\textnormal{\textbf{CSR}} < \tau$ 
    %\del{priority $pkt.\textnormal{\textbf{Pri}} > 0$}
    is immediately discarded, indicating that no discarding-related information has been embedded into the previous packets in current frame. 
    Consequently, the video analytics server cannot distinguish between packet loss caused by \SYS and other unexpected cases (\eg AQM mechanisms \cite{ryu2004advances, hollot2002analysis} or channel interferences \cite{kihero2021wireless, villegas2007effect}). 
  %  \vspace{-1em}
    \item Moreover, as shown in Figure \ref{fig:Design-7}, since the potential discarding threshold may be higher or lower than the current packet's CSR, a uniform execution of packet discarding is not feasible for those two cases.
\squishend

Therefore, line 5 determines whether the current packet is the first packet of the frame with \texttt{Check1stPktInFrame} by checking if (1) its CSR equals 100\% (encoded as 127), indicating that it belongs to the block with the highest importance (priority=15), and (2) it contains the block start indicator, as expressed by the following two conditions:\vspace{-1.5mm}
\begin{equation}
    \label{Equ:checkFirstPacketInFrame}
    pkt.\textnormal{\textbf{CSR}} == 127\quad\textnormal{and}\quad pkt.\textnormal{\textbf{S}} == True.\vspace{-1.5mm}
\end{equation}
When Equation (\ref{Equ:checkFirstPacketInFrame}) is met, the \textbf{PRE} phase transitions to \textbf{PD}.

\begin{figure}[t!]
    \centering
    \includegraphics[width=0.8\linewidth]{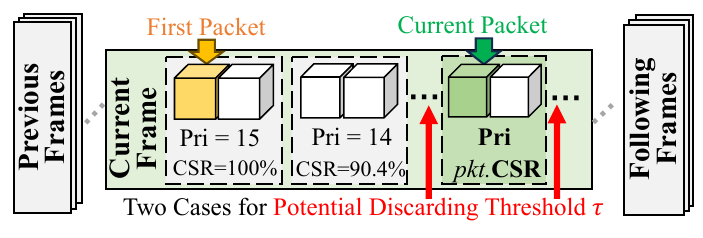}
    \vspace{0mm}
    \caption{Potential discarding threshold in \textbf{PRE} phase.}
    \vspace{2mm}
    \label{fig:Design-7}
\end{figure}

(3) \textbf{PD} for packet discarding (\textit{lines 7-15}):
In the \textbf{PD} phase, the packet discarding manager makes discarding decisions based on the packet discarding threshold $\tau$, which was calculated in the \textbf{INIT} phase. 
Specifically, for each packet, the packet discarding manager compares its priority value $\mathit{pkt}.\textnormal{\textbf{CSR}}$ with $\tau$.

At an edge router, each packet is compared against this threshold: packets whose $CSR \geq \tau$ are retained and their C bit is set to indicate proactive discarding in the frame (\textit{line 10-11}), while packets whose $CSR < \tau$ are subject to dropping. 

To make discarding decisions explicit, Algorithm~\ref{Alg:Design-1} marks a \emph{dropping boundary}. 
Because the client transmits packets in descending order of CSR, the router encounters them approximately in importance order. 
Instead of attempting to identify the last above-threshold packet\textemdash which would require knowledge of future arrivals\textemdash the procedure \texttt{Check1stPktNextBlk} designates the first block-start packet ($S{=}1$) with $\mathit{CSR} < \tau$ as the boundary. 
This packet is retained and flagged with $B$, and all subsequent packets in the frame are discarded. 
Boundary marking allows the video analytics server to decode the frame without unnecessary retransmission requests (see \S\ref{subsubsec:header-ext}).

\del{By combining threshold-based discarding with explicit boundary marking, Algorithm~\ref{Alg:Design-1} enables precise and stateless packet dropping, without introducing per-flow state, while ensuring that video analytics can continue seamlessly even under severe bandwidth deficits.}

% If it exceeds $Dis$, the packet is discarded with $DiscardPacket$ (as shown in lines 12-15).
% Otherwise, the packet discarding manager embeds the discarding priority into the surviving packet as $pkt.Dis$, assisting the video analytics server with depacketization and decoding the frame without triggering unnecessary retransmission requests. 
% More details on retransmission avoidance can be found in \S\ref{subsubsec:header-ext}. 
%More details on how the video server avoids retransmission using the discarding information can be found in \S\ref{subsubsec:header-ext}.

%To summarize, Figure \ref{fig:Design-4} illustrates the state transitions for proactive packet discarding at the bottleneck router between \textbf{INIT}, \textbf{PRE}, and \textbf{PD}. Notably, when the state is either \textbf{PRE} or \textbf{PD}, it indicates that a sudden and sharp drop in available bandwidth is currently occurring.

%\subsubsection{Discarding Priority-Aware Bitrate Adjustment}
%In coordination with the bottleneck router's proactive packet discarding, the video analytics server dynamically adjusts the target bitrate, which is then communicated to the sender.

\MyPara{Multiple Video Flows.} 
When multiple video analytics flows share the same edge link, \SYS must balance proactive discarding across flows. 
If the edge router already maintains per-flow priority queues\textemdash a capability supported by many edge routers~\cite{rfc8290, cisco_wfq_cfg_2008,cisco_cbwfq_feat_12_0t, juniper_cos_fc_overview,juniper_cos_custom_fc}, fairness can be naturally enforced by applying \SYS within each flow's queue independently. 
In this setting, each flow experiences discarding decisions aligned with its own CSR encoding, while the underlying scheduler ensures fair capacity allocation across the flows.
When per-flow queues are unavailable, however, the router can no longer guarantee per-flow isolation. 
In such cases, our objective shifts to maximizing the \emph{average analytics accuracy} across all flows, ensuring that the shared link is utilized in a way that best preserves global application-level performance.

%% file: 6-evaluation.tex
\mysection{Evaluation
\label{sec:evaluation}}
%In this section, we first describe the experimental setup and then present the evaluation results of \SYS.
%\mysubsection{\SYS Implementation}

%\begin{figure}[t!]
%    \centering
%    \includegraphics[width=0.9\linewidth]{figure/Implementation.pdf}
%    \vspace{0mm}
%    \caption{Overview of \SYS implementation. \hx{Should this overview go to the front (beginning of \S\ref{sec:design})}?}
%    \vspace{0mm}
%    \label{fig:Implementation}
%\end{figure}


\begin{figure*}[!t]
    \vspace{0mm}    
    \centering
    \subfigbottomskip=1mm %设置第二行子图与第一行子图的距离，即下面的头与上面的脚的距离
    \subfigcapskip=0mm %设置子图与子标题之间的距离
    % 第一行
    \subfigure[PCTL frame delays, 99.9\% (ms): \newline (\textbf{152}, 301, 572)]{
        \includegraphics[width=0.235\linewidth]{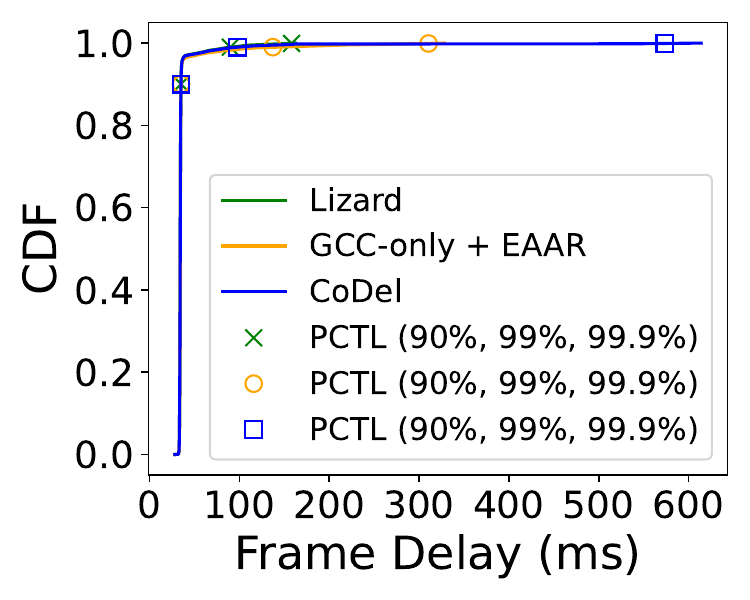}
        \label{fig:Eva_overall_1}
    }\hfill
    \subfigure[PCTL accuracy on Auburn, 0.1\%: \newline (\textbf{0.46}, 0.40, 0.40, 0.43)]{
        \includegraphics[width=0.235\linewidth]{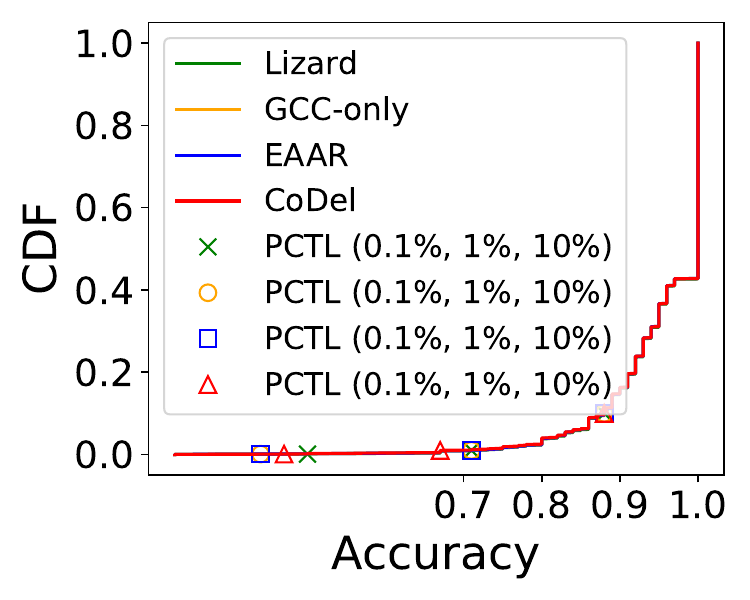}
        \label{fig:Eva_overall_2}
    }\hfill
    \subfigure[PCTL accuracy on Banff, 0.1\%: \newline  (\textbf{0.76}, 0.72, 0.72, 0.68)]{
        \includegraphics[width=0.235\linewidth]{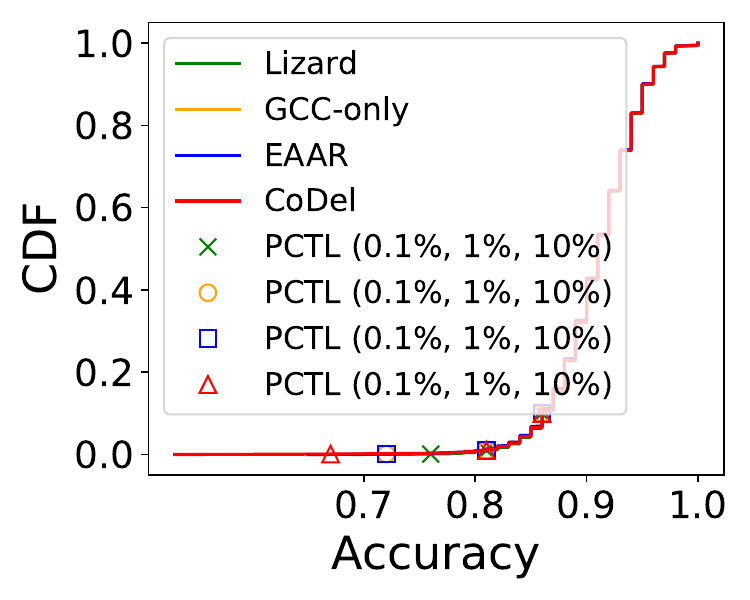}
        \label{fig:Eva_overall_3}
    }\hfill
    \subfigure[PCTL accuracy on Jacksonhole, 0.1\%:  (\textbf{0.80}, 0.50, 0.58, 0.58)]{
        \includegraphics[width=0.235\linewidth]{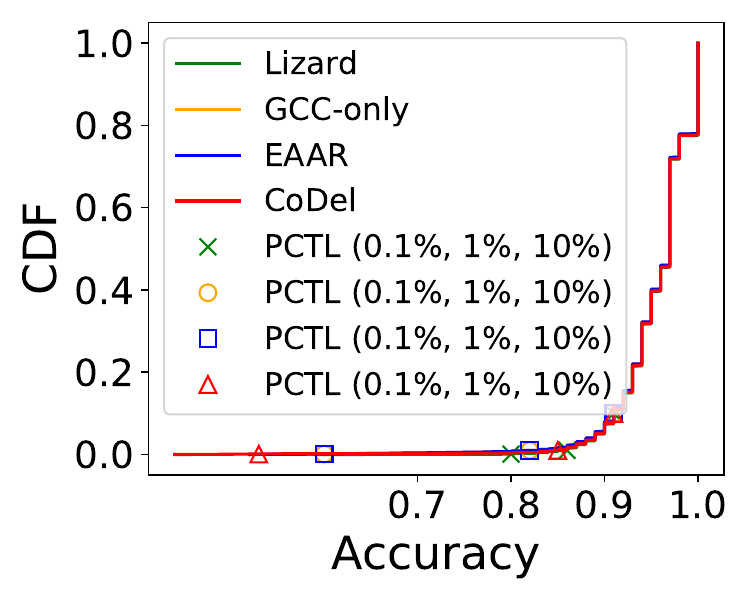}
        \label{fig:Eva_overall_4}
    }
    \vspace{0mm}
    \caption{Evaluated with network dataset Oboe, using different video datasets Auburn, Banff and Jacksonhole.
    \label{fig:Eva_1_2_3_4}
    %\shan{It's hard for me to identify the results of Lizard and GCC-only, as the line of CoDel covers all others.}\yu{OK. I will add some explanations in body text.}
    }
    \vspace{-4mm}
\end{figure*}

\begin{figure*}[!t]
    \vspace{0mm}
    \centering
    \subfigbottomskip=0mm %设置第二行子图与第一行子图的距离，即下面的头与上面的脚的距离
    \subfigcapskip=0mm %设置子图与子标题之间的距离
    % 第一行
    \subfigure[PCTL frame delays, 99.9\% (ms): \newline  (\textbf{337}, 591, 989)]{
        \includegraphics[width=0.235\linewidth]{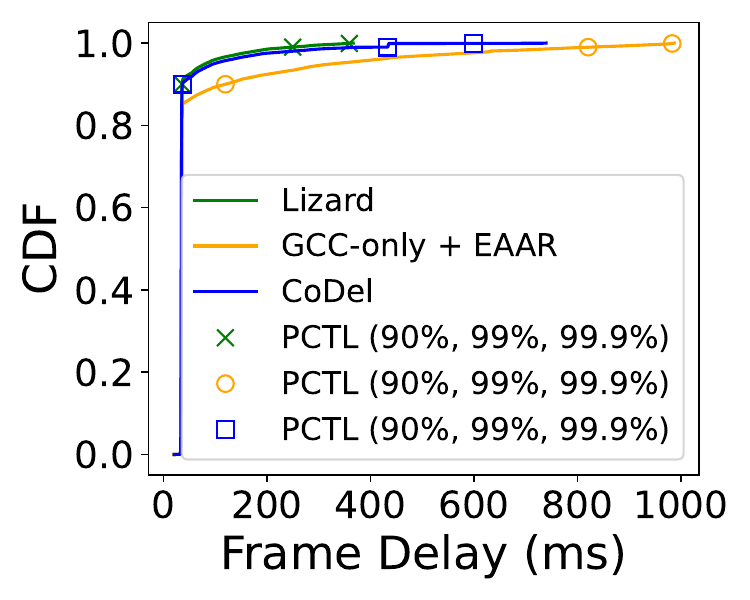}
        \label{fig:Eva_overall_5}
    }\hfill
    \subfigure[PCTL accuracy on Auburn, 0.1\%: \newline (\textbf{0.61}, 0.51, 0.51, 0.51)]{
        \includegraphics[width=0.235\linewidth]{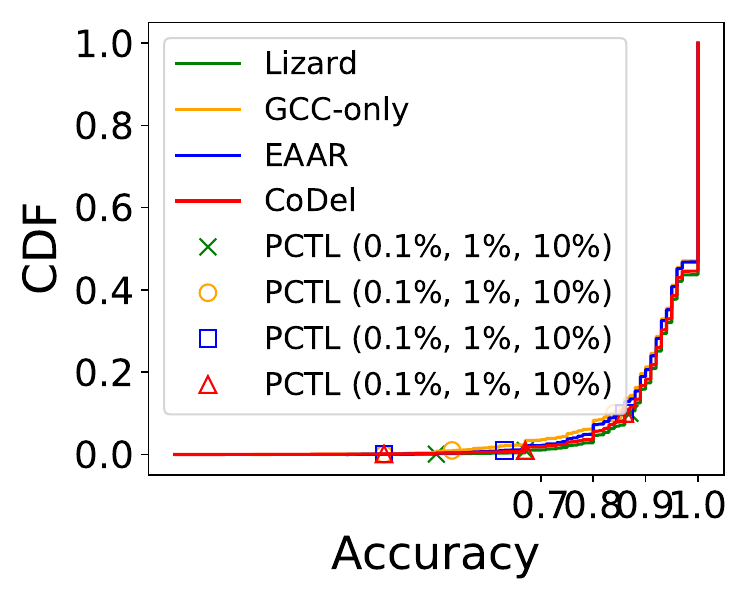}
        \label{fig:Eva_overall_6}
    }\hfill
    \subfigure[PCTL accuracy on Banff, 0.1\%: \newline (\textbf{0.75}, 0.62, 0.58, 0.68)]{
        \includegraphics[width=0.235\linewidth]{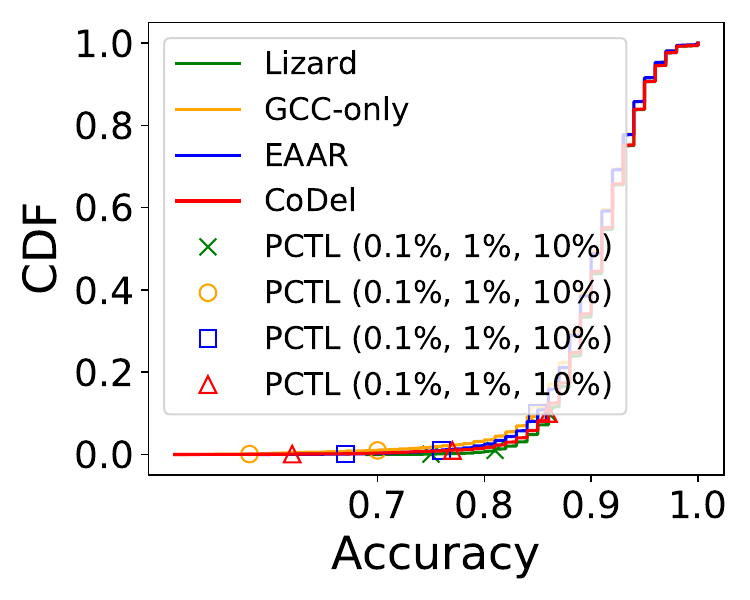}
        \label{fig:Eva_overall_7}
    }\hfill
    \subfigure[PCTL accuracy on Jacksonhole, 0.1\%: (\textbf{0.78}, 0.72, 0.58, 0.48)]{
        \includegraphics[width=0.235\linewidth]{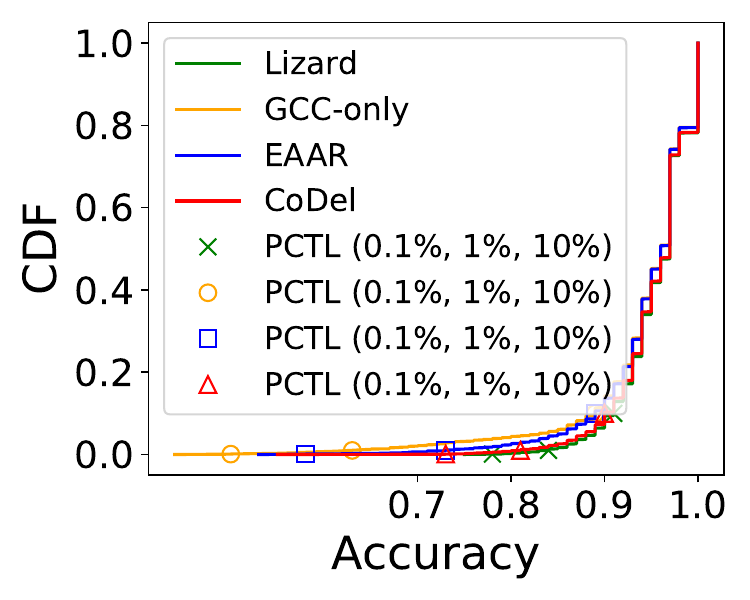}
        \label{fig:Eva_overall_8}
    }
    \vspace{0mm}
    \caption{Evaluated with network dataset Ghent, using different video datasets Auburn, Banff and Jacksonhole.}
    \label{fig:Eva_5_6_7_8}
    \vspace{-4mm}
\end{figure*}

%\subsection{Experiment Setup}
\MyPara{Implementation.} We set up an edge-based real-time video analytics system on a Mininet testbed \cite{kaur2014mininet}, running Ubuntu 20.04 with an Intel Haswell CPU and an NVIDIA Tesla T4 GPU. The testbed, based on the topology shown in Figure \ref{fig:Motivation-Background}, consists of three network nodes: a client, an edge router, and a video analytics server\textemdash a typical setup in edge networks \cite{wang2020joint, liu2019edge}.
RTP/RTCP-based real-time video packet transmission was implemented using aiortc \cite{aiortc}, and video analytics implemented on PyTorch \cite{paszke2019pytorch}. A detailed description of \SYS's implementation can be found in \apx{\S\ref{appendix:detailed_implementation}}{the extended version~\cite{lizardExtended}}.

\MyPara{Datasets.}
We evaluated \SYS using two real-world network bandwidth datasets: Oboe \cite{Oboe} and Ghent \cite{Ghent} where the distribution of the length and frequency of ABW reductions is reported in Table \ref{tab:ABW_drop_frequency}.
We used four real-world video datasets including three from surveillance cameras (\ie Auburn \cite{auburndata}, Banff \cite{banff}, and Jacksonhole \cite{jackson}), as well as a mobile AV camera dataset nuScenes \cite{nuscenes2019}. 
%Table \ref{tab:video_dataset} reports their statistics. \rn{if tight on space, you can remove the previous sentence+table}
%For video analytics, we use the video datasets Auburn \cite{auburndata}, Banff \cite{banff}, and Jacksonhole \cite{jackson}, with their basic descriptions provided in Table \ref{tab:video_dataset}.
%\hx{I am here.}

\MyPara{Metrics for video analytics.}
In our experiments, we focused on the task of object detection \cite{zou2023object}, a fundamental aspect of video analytics.
 Specifically, we leveraged the object detection model YOLOX \cite{ge2021yolox} to target two key detection tasks: vehicle detection and pedestrian detection. The number of blocks N was configured to 4 to run \SYS. A detailed evaluation of the impact of N will be discussed shortly in \S\ref{subsubsec:parameterAndOverhead}.
 
The evaluation of accuracy for these tasks was conducted as follows.
In edge-based real-time video analytics, the client receives the analysis results with a certain delay.
Thus, these results must be compared with the ground truth from subsequent frames.
For each frame $i$, let $L_{i}^{\textnormal{Det}}$ denote the detection results, a list of bounding boxes.
The ground truth detection results for frame $i$ are represented by $L_{i}^{\textnormal{GT}}$.
For every bounding box $bbx^{\textnormal{Det}}$ in $L_{i}^{\textnormal{Det}}$, we assess whether there is a corresponding bounding box in $L_{i}^{\textnormal{GT}}$ such that their Intersection Over Union (IOU) \cite{rezatofighi2019generalized} exceeds a predefined threshold (\eg 85\%).
If such a bounding box is found, $bbx^{\textnormal{Det}}$ is classified as a True Positive (TP) \cite{lepeschkin1958characteristics}.
Let $tp_i$ denote the count of TP bounding boxes in frame $i$.
The accuracy using F1-score \cite{goutte2005probabilistic} is:%\vspace{-1mm}
% \begin{equation}
% \footnotesize Acc = \mathbb{E}_{i}\left[\frac{2 \cdot \textnormal{precision} \cdot \textnormal{recall}}{\textnormal{precision} + \textnormal{recall}}\right] = \mathbb{E}_{i}\left[\frac{2 \cdot tp_i}{\textnormal{len}(L_{i}^{\textnormal{Det}}) + \textnormal{len}(L_{i}^{\textnormal{GT}})}\right],\vspace{-1mm}
% \end{equation}
\begin{equation}
\footnotesize
Acc = \mathbb{E}_{i}\left[\frac{2 \cdot \textnormal{precision} \cdot \textnormal{recall}}{\textnormal{precision} + \textnormal{recall}}\right] 
= \mathbb{E}_{i}\left[\frac{2 \cdot tp_i}{\textnormal{len}(L_{i}^{\textnormal{Det}}) + \textnormal{len}(L_{i}^{\textnormal{GT}})}\right],
\end{equation}

\vspace{-1mm}

where $\textnormal{len}(L)$ represents the number of bounding boxes in $L$, and $\mathbb{E}[\cdot]_i$ denotes the average value across all frames.
This evaluation method is applied to both vehicle and pedestrian detection tasks to assess their respective performances. 
Additionally, we measured the frame delay as the time interval from when a frame is encoded on the client to when it is reassembled on the server and becomes ready for processing by the video analytics model.
Based on related work \cite{meng2022achieving, meng2023enabling}, tail latency and its impact on network system performance are critical. 
Therefore, in our experimental results, we also include metrics for tail frame delay (\eg 99\% and 99.9\%) and tail accuracy (\eg 1\% and 0.1\%).

\MyPara{Baselines.}
We compared \SYS with these baselines, corresponding to client-based, router-based, and server-based methods to enhance edge-based real-time video analytics:\vspace{-0.25em}
\squishlist
    \item GCC-only \cite{carlucci2016analysis} relies solely on the client's bitrate controller GCC to passively adjust the real-time video bitrate in order to reduce the queue delay at the last-mile router. However, this reactive bitrate adjustment introduces a lag of at least one RTT, as analyzed in \S\ref{subsec:motivation2}, which can degrade real-time video analysis performance. \vspace{-0.25em}
    \item CoDel \cite{sharma2014controlling} discards packets based on the network condition of queuing delay at the last-mile router. However, it is limited by the interdependence of packets within a video frame and can compromise the timeliness of real-time video analytics due to packet retransmission. \vspace{-0.25em}
    \item EAAR \cite{liu2019edge} detects delayed frames (\eg delay $> 400\,\mathrm{ms}$) on the server and uses a motion vector-based prediction approach \cite{liu1993new} to obtain video analytics results for them. Its predictive performance depends on the rate of change in the video content over time, which is inherently uncertain. \vspace{-1em}
\squishend
%\shan{Maybe add more explanations on the comparisons between the baselines and our method. How does one differ from another? What are their strengths/weakness?}
%\yu{I will add some explanations}

%\mysubsection{Experimental Results
%\label{subsec:ExperimentResults}}
\mysubsection{Overall Performance
\label{subsubsec:overallperformance}}
%\MyPara{Overall results.}
We present the results of frame delay and accuracy for \SYS compared to three baselines across varying network bandwidth traces, video datasets, and video analysis tasks, as shown in Figures \ref{fig:Eva_1_2_3_4}, \ref{fig:Eva_5_6_7_8}, and \ref{fig:Eva_9_10}. 
We also report the actual 99.9\% latency and 0.1\% accuracy numbers in each caption for clarity, and highlighted are \SYS’s performance in Figures \ref{fig:Eva_1_2_3_4} and \ref{fig:Eva_5_6_7_8}. \rev{In the frame-delay CDFs, GCC-only and EAAR share one curve because EAAR changes only server-side prediction, not the transport, so their frame delays are identical while their accuracies differ.}
% Overall, \SYS surpasses the baselines, reducing the 99.5th percentile frame delay by 53.2\% and enhancing accuracy by up to 27.1\%. \rev{The gains concentrate in the tail, which full-range CDFs compress visually, so the captions list the tail values: the 99.9th-percentile frame delay is 152 ms for \SYS versus 301 ms (GCC-only+EAAR) and 572 ms (CoDel) on Oboe, and 337 versus 591 and 989 ms on Ghent, a 2--3$\times$ reduction (see also the drill-down CDFs in \S\ref{sec:drilldown}). The tail is where real-time analytics breaks: medians are similar across schemes, but Table~\ref{tab:frame_delay_accuracy} shows accuracy degradation growing from 4.70\% to 23.51\% (person) and 2.73\% to 27.38\% (vehicle) as frame delay rises from 50 to 800 ms, and SLA violations are incurred on precisely these tail frames~\cite{yan2022resource, xu2024adaptive}.}
% Regarding the frame delay, CoDel suffers from a significant number of retransmissions, while both GCC-only and EAAR exhibit significant lags in responding to drastic ABW changes. 
Overall, \SYS outperforms all baselines, reducing the 99.5th-percentile frame delay by 53.2\% and improving accuracy by up to 27.1\%. The gains are concentrated in the tail: at the 99.9th percentile, \SYS reduces frame delay from 301 to 152~ms on Oboe and from 591 to 337~ms on Ghent compared with GCC-only+EAAR, with even larger gains over CoDel (572 to 152~ms and 989 to 337~ms, respectively). This tail reduction is critical for real-time analytics, where delayed frames incur substantially larger accuracy loss and SLA violations (see Table~\ref{tab:frame_delay_accuracy}). Regarding the frame delay, CoDel is mainly penalized by retransmissions, while GCC-only and EAAR exhibit significant lags in responding to drastic ABW changes.

In terms of accuracy, the high frame delays in GCC-only and CoDel compromise the real-time nature of video analysis, thereby reducing the accuracy of the results. 
Although EAAR can leverage the motion vector \cite{liu2019edge} to compute video analytics results based on historical data when the frame delay is high (\eg 400 ms), its predictive accuracy in the temporal dimension is limited.
Consequently, it underperforms \SYS, which discards less important blocks in the spatial dimension to optimize performance.
\apx{A detailed analysis of \SYS's time-series performance can be found in \S\ref{sec:drilldown}.}{}

\MyPara{Various network bandwidth traces.}
This experiment evaluates the performance of \SYS on different network bandwidth datasets, Oboe and Ghent. 
As shown in Table \ref{tab:ABW_drop_frequency}, we observe that the Ghent dataset exhibits more frequent and severe ABW degradations compared to Oboe. 
Therefore, the network conditions reflected in the Ghent dataset present a greater need for the application of \SYS and offer more potential for performance improvement.
We analyze the experimental results presented in Figures \ref{fig:Eva_1_2_3_4} and \ref{fig:Eva_5_6_7_8} and observe that on the Oboe dataset, the CDFs of \SYS and the baselines are relatively close due to a relatively small number of sharp ABW reductions, while the Ghent dataset contains many more ABW drops and, hence, the gap between \SYS and the baselines' CDFs is more pronounced. 
Overall, \SYS consistently outperforms the baselines in terms of 0.1th percentile accuracy. 
Specifically, the 0.1th percentile accuracy improvements of \SYS over the baselines are 12\% for the Oboe dataset and 19\% for the Ghent dataset. 
% These results underscore the explainable differences in \SYS's performance improvements across different network datasets.

\begin{figure}[!t]
    \vspace{0mm}
    \centering
    \subfigbottomskip=0mm %设置第二行子图与第一行子图的距离，即下面的头与上面的脚的距离
    \subfigcapskip=0mm %设置子图与子标题之间的距离
    \subfigure[Frame delay and accuracy.]{
        \includegraphics[width=0.47\linewidth]{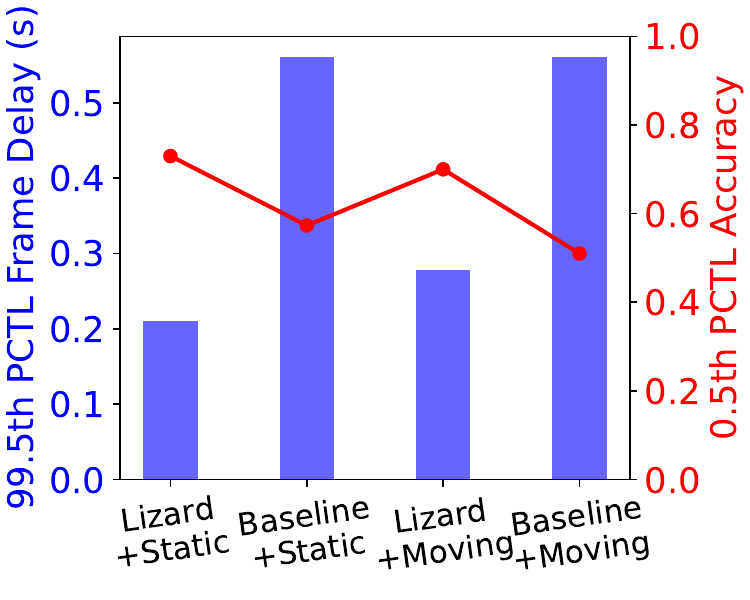}
        \label{fig:Eva_overall_MobileCamera_1}
    }\hfill
    \subfigure[Varying feedback intervals.]{
        \includegraphics[width=0.47\linewidth]{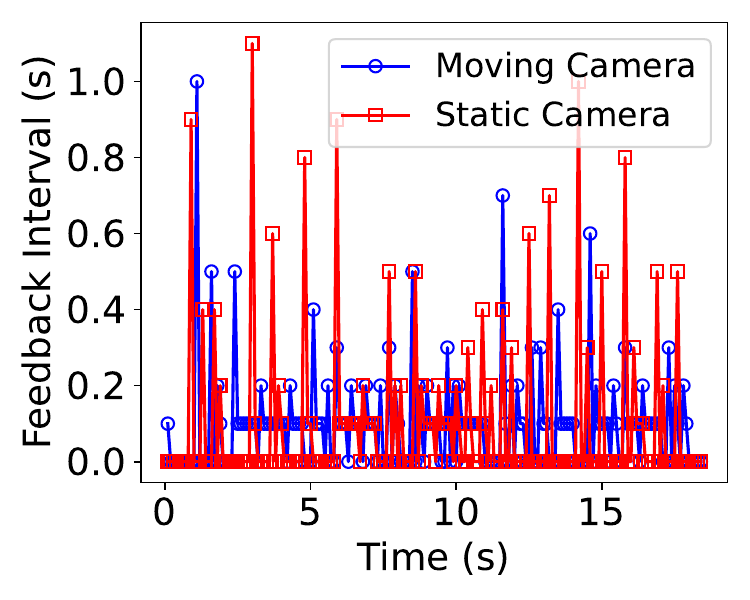}
        \label{fig:Eva_overall_MobileCamera_2}
    }
    \caption{Evaluated with moving and static camera videos.}
    \label{fig:Eva_overall_MobileCamera}
    \vspace{0mm}
\end{figure}

\MyPara{Various video datasets.}
Furthermore, we evaluated the performance of \SYS using different video datasets. 
We observe that the Auburn, Banff, and Jacksonhole datasets differ in terms of the proportion of the area occupied by the target objects, with Auburn and Banff having relatively lower proportions and Jacksonhole having a higher one. 
In other words, the Jacksonhole dataset provides \SYS with a larger decision space for proactive packet dropping, allowing it to maximize its benefits, while the Auburn and Banff datasets offer a more constrained decision space. 
We analyzed the results from Figures \ref{fig:Eva_overall_2}, \ref{fig:Eva_overall_3} and \ref{fig:Eva_overall_4} across these datasets and calculate the 0.1th percentile accuracy improvements of \SYS on Auburn, Banff, and Jacksonhole, which are 16\%, 13\%, and 38\%, respectively. 
%They are consistent with our previous analysis and support the rationale behind it.

In addition to the static camera videos above, we evaluated \SYS on moving camera videos (\ie nuScenes) in Figure \ref{fig:Eva_overall_MobileCamera}.
Moving camera videos generally have more dynamic content and fewer discardable regions that do not contain interesting targets. 
As shown in Figure \ref{fig:Eva_overall_MobileCamera_1}, under the same network conditions (using Oboe and Ghent), GCC-only (the baseline) incurs a 99.5th percentile frame delay of nearly 600 ms for both video types. 
\SYS reduces this delay to $\thicksim$200 ms for static camera videos by discarding low-importance packets. 
For moving camera videos, \SYS slightly increases the 99.5th percentile frame delay to 278 ms compared to static camera videos, as fewer regions are safely discardable. 
This reduces the accuracy to 0.69, as compared to 0.73 for static videos.
% As a result, the 0.5th percentile accuracy for moving camera videos drops from 0.73 to 0.69. 
However, \SYS improves accuracy by approximately 0.1 over the baseline for both video types.
Figure \ref{fig:Eva_overall_MobileCamera_2} shows the variation in feedback intervals over time for \SYS’s priority-based feedback mechanism. 
A value of 0 indicates no feedback is needed, while values above 0 reflect the time since the last feedback. 
For static camera videos, most data points are near 0 or higher intervals, indicating infrequent updates. 
In contrast, moving camera videos experience frequent priority changes, prompting \SYS to adaptively increase feedback frequency to ensure up-to-date priority data for packet discarding.

\begin{figure}[!t]
    \vspace{0mm}
    \centering
    \subfigbottomskip=0mm %设置第二行子图与第一行子图的距离，即下面的头与上面的脚的距离
    \subfigcapskip=0mm %设置子图与子标题之间的距离
    \subfigure[Task of vehicle detection.]{
        \includegraphics[width=0.47\linewidth]{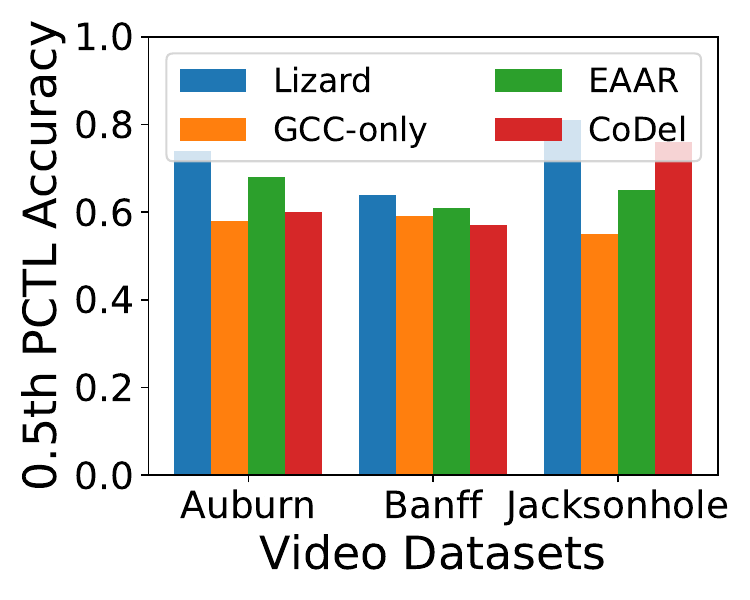}
        \label{fig:Eva_overall_9}
    }\hfill
    \subfigure[Task of person detection.]{
        \includegraphics[width=0.47\linewidth]{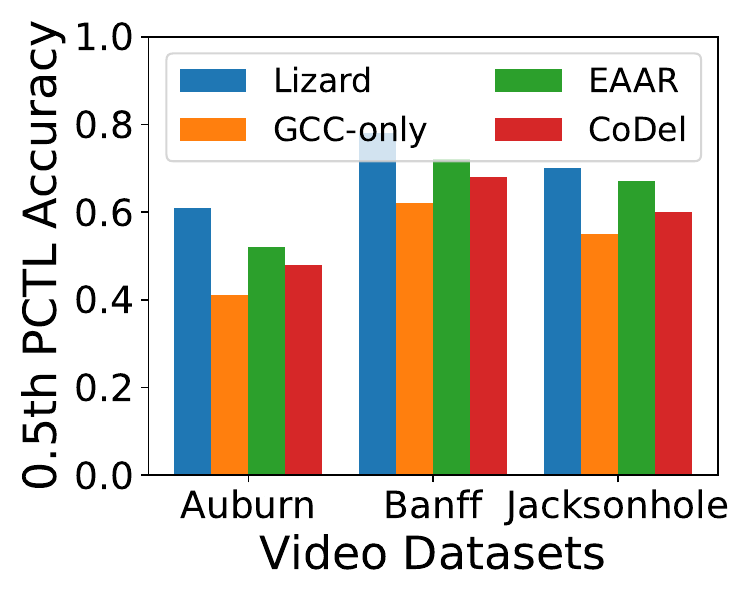}
        \label{fig:Eva_overall_10}
    }
    \caption{Evaluated with various video analytics tasks.}
    \label{fig:Eva_9_10}
    \vspace{0mm}
\end{figure}

\MyPara{Various video analytics tasks.}
To measure the impact of varying target occupation ratios in different video datasets on \SYS's performance improvement, we conducted a fine-grained evaluation for vehicle and person detection tasks. 
The experimental results are presented in Figure \ref{fig:Eva_9_10}. 
Specifically, we calculate the performance improvements of \SYS over the baseline for both vehicle and person detection tasks across the different video datasets. 
We find that on the Jacksonhole dataset, \SYS's 0.5th percentile accuracy improvement for the vehicle detection task (\ie 24\%) is greater than that for the person detection task (\ie 17\%). However, on the Auburn and Banff datasets, \SYS's accuracy improvements for the vehicle detection task (\ie 19\% and 8\%, respectively) are lower than those for the person detection task (\ie 27\% and 15\%, respectively).
An inspection found that, on the Jacksonhole dataset, vehicle targets are more clustered and predictable compared to person targets. 
This allows \SYS to achieve greater performance gains through proactive packet dropping. 
Similarly, the distribution characteristics of person targets in the Auburn and Banff datasets provide \SYS with more opportunities to enhance the performance of real-time video analysis for person detection.

\begin{figure}[!t]
    \vspace{0mm}
    \centering
    \subfigbottomskip=0mm %设置第二行子图与第一行子图的距离，即下面的头与上面的脚的距离
    \subfigcapskip=0mm %设置子图与子标题之间的距离
    \subfigure[Accuracy impacted by block sizes.]{
        \includegraphics[width=0.46\linewidth]{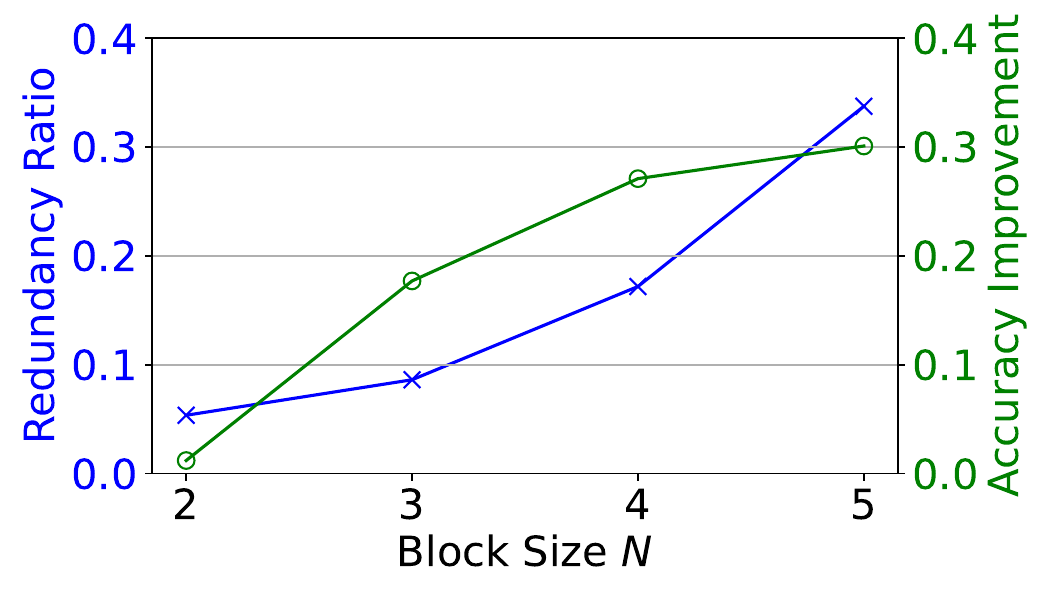}
        \label{fig:Eva_parameter_bsize}
    }\hfill
    \subfigure[Impact of parameter $\theta$.]{
        \includegraphics[width=0.46\linewidth]{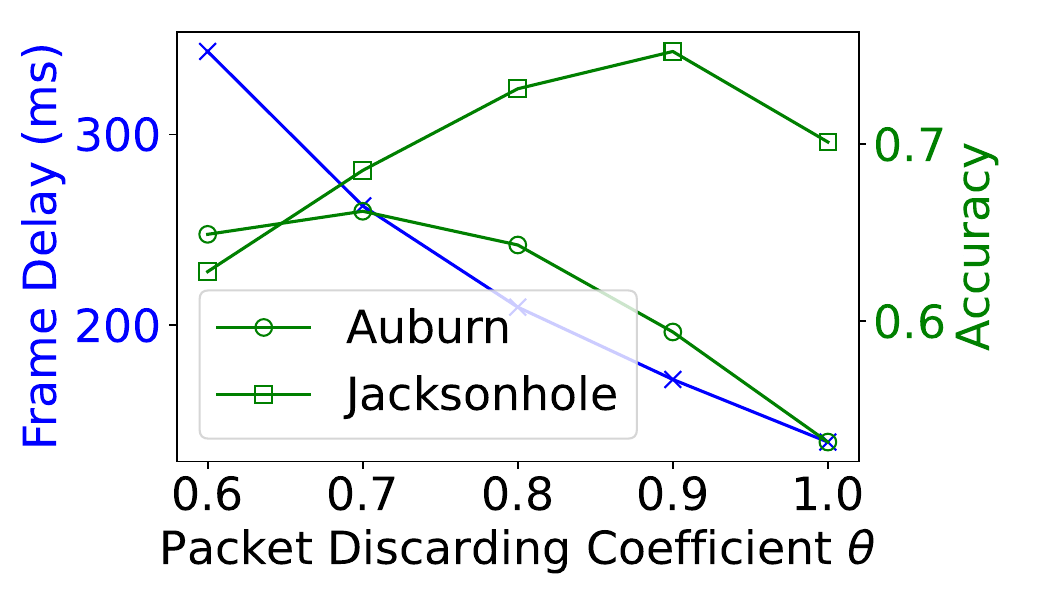}
        \label{fig:Eva_parameter_11_12}
    }\\
    \subfigure[Time cost of each module.]{
        \includegraphics[width=0.46\linewidth]{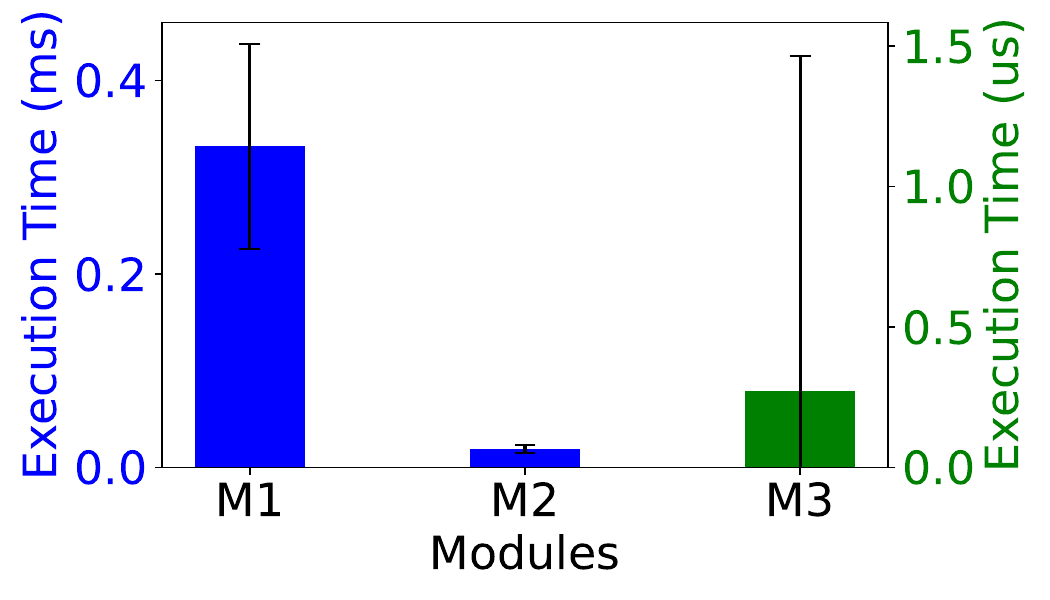}
        \label{fig:Eva_overhead_13}
    }\hfill
    \subfigure[Time cost of packet discarding.]{
        \includegraphics[width=0.46\linewidth]{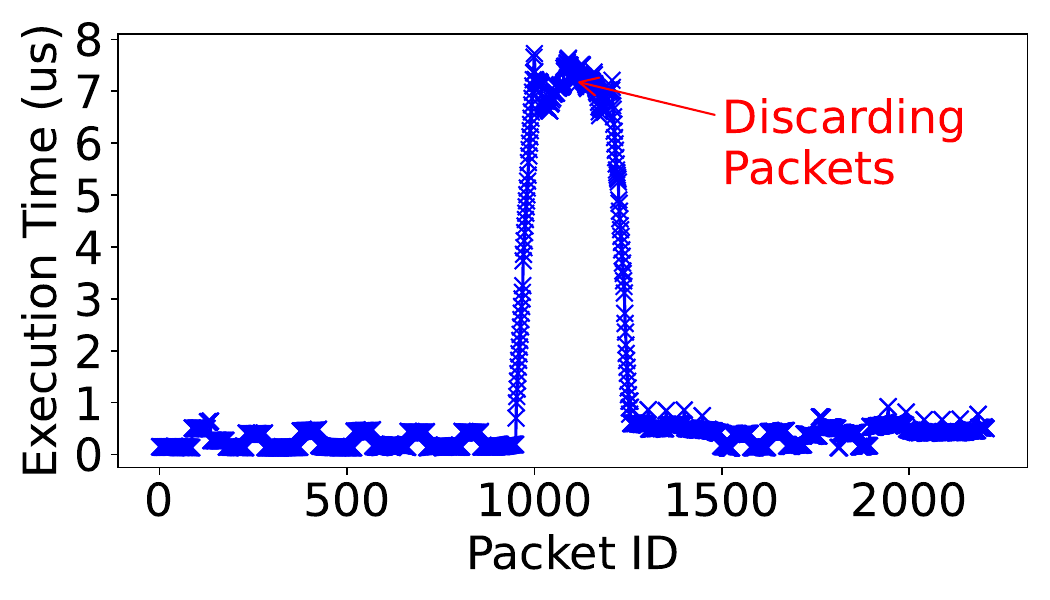}
        \label{fig:Eva_overhead_14}
    }
    \caption{Parameter impact and system overhead.}
    \label{fig:Eva_parameter_overhead}
    \vspace{0mm}
\end{figure}

\mysubsection{Parameter Impact and Overhead
\label{subsubsec:parameterAndOverhead}}
\MyPara{Impact of block size selection.}
We investigated the impact of block size on data redundancy and accuracy in the presence of ABW drops; their results are reported in Figure \ref{fig:Eva_parameter_bsize}. The left Y-axis represents the data redundancy ratio introduced by each block configuration while the right Y-axis represents the accuracy improvements over the baseline (GCC-only). 
Given this experiment focuses on the impact of N, we developed a simulation that allows us to evaluate with a larger N (\eg beyond the upper bound supported by \SYS, which is 4).
When coarse-grained blocks are used (\eg $N = 2$), discarding blocks can result in the loss of blocks containing important objects, leading to a minimal accuracy improvement of 0.012. Increasing N from 2 to 5 improves the accuracy when ABW drops as finer-grained blocks allow \tool to discard less critical blocks. However, the accuracy improvements come at the cost of increased redundancy\textemdash the redundancy ratio reaches up to more than 30\% when the block size is set to 5. Our evaluation suggests that N=4 strikes the right balance.

\MyPara{Impact of packet discarding coefficient.}
The delay mitigation parameter $\theta$, defined in Equation (\ref{equ:delayMitigation}), can be used as the packet discarding coefficient. 
% This coefficient can be adjusted on demand to flexibly control the packet discard ratio in response to sudden drops in network bandwidth. 
We varied $\theta$ from 0.6 to 1.0 and applied it to different video datasets. 
The 99.5th percentile frame delay and accuracy were then evaluated under various conditions, with the results presented in Figure \ref{fig:Eva_parameter_11_12}.
These results highlight two major takeways:
\textit{(1) The packet discarding coefficient $\theta$ is not always better when increased.} 
For both the Auburn and Jacksonhole datasets, increasing $\theta$ can effectively alleviate queuing delays at the router, thereby reducing frame delay. 
However, once $\theta$ exceeds a certain threshold, accuracy begins to decline. 
This is because discarding too many packets negatively impacts video analytics accuracy.
\textit{(2) The optimal value of the packet discarding coefficient $\theta$ varies across different video datasets.}
We observe that the optimal $\theta$ for Auburn and Jacksonhole is not the same. 
The optimal $\theta$ for Auburn (\ie 0.7) is lower than that for Jacksonhole (\ie 0.9). 
This difference aligns with our previous analysis, which indicates that the distribution of important information (\eg the detected objects) in video frames varies across different video datasets.

\MyPara{Time components of \SYS.}
We measure the runtime overhead of \SYS's three main modules: block-aware RTP header extension (M1), priority-based video analytics feedback (M2), and phase-transition-based packet discarding (M3), as shown in Figure~\ref{fig:Eva_overhead_13}. M1 and M2 run at the client and server, respectively, at the frame level; both take less than 1~ms, negligible compared with the 33~ms frame interval at 30~fps. M3 runs at the router at the packet level, averaging 0.27~$\mu$s per packet with a standard deviation of 1.19~$\mu$s. Figure~\ref{fig:Eva_overhead_14} explains this variance: proactive dropping adds about 7~$\mu$s per packet, while the non-proactive phase incurs little overhead. \apx{\S\ref{sec:drilldown} further evaluates \SYS's overhead for a specific analytics task.}{}

%% file: 2-relatedwork.tex
\mysection{Related Work
\label{sec:relatedwork}}
\MyPara{Client-side optimizations.}
Client-side optimizations~\cite{wang2020joint, zhang2021adaptive, shi2023adapyramid, li2020reducto, xie2019source, du2022accmpeg} reduce latency and bandwidth consumption by processing or adapting video close to the source. Wang \textit{et al.}~\cite{wang2020joint} optimize configuration choices under bandwidth constraints, while Zhang \textit{et al.}~\cite{zhang2021adaptive} extend this approach to multi-edge settings. AdaPyramid~\cite{shi2023adapyramid} dynamically adapts configurations within each high-resolution frame; Reducto~\cite{li2020reducto} filters frames on camera before transmission; and AccMPEG~\cite{du2022accmpeg} optimizes video encoding for low latency and high DNN accuracy. More recent systems further reduce transmitted data through on-device frame gating~\cite{hu2015offload, iyengar2023offload}, content-aware block-level quantization~\cite{zhou2023bandwidth}, and filtering of redundant spatial-temporal semantics~\cite{chen2023edge}.
These techniques are proactive with respect to \emph{content}, but remain reactive to \emph{bandwidth}: endpoints detect congestion only after queues form at the router (\S\ref{subsec:motivation2}). \SYS complements them by exposing content priorities to the network, allowing clients to transmit low-priority frames or blocks when bandwidth permits and letting the router discard them only when ABW drops. This preserves quality at high bandwidth while adapting quickly to sudden degradation.
\begin{comment}
Client-side optimizations \cite{wang2020joint, zhang2021adaptive, shi2023adapyramid, li2020reducto, xie2019source, du2022accmpeg} focus on reducing latency, conserving bandwidth, and distributing computational load by processing data closer to the source.
Wang \textit{et al.} \cite{wang2020joint} optimize configuration choices under limited bandwidth, while Zhang \textit{et al.} \cite{zhang2021adaptive} extend this to a multi-edge scenario.
Shi \textit{et al.} \cite{shi2023adapyramid} propose AdaPyramid, which dynamically adapts configurations within each frame for high-resolution video, improving speed and accuracy.
Li \textit{et al.} \cite{li2020reducto} introduce Reducto, applying on-camera filtering to reduce data processing and transmission.
%Xie \textit{et al.} \cite{xie2019source} present a compression technique for IoT edge devices, balancing perception loss and efficiency.
Du \textit{et al.} \cite{du2022accmpeg} offer AccMPEG, optimizing video encoding for reduced latency and high DNN accuracy.
While they can adapt to drastic ABW degradation, they are inherently reactive. 
As previously discussed in \S\ref{subsec:motivation2}, these client-side methods exhibit significant lag, rendering them inadequate for the ultra-low-latency requirements of edge-based real-time video analytics.
\end{comment}

\MyPara{Server-side optimizations.}
Server-side optimizations \cite{canel2019scaling, yuan2022infi, yuan2023accdecoder, mi2024accelerated, han2015deep, yuan2023packetgame} improve video analytics by leveraging edge server resources to efficiently handle computational demands, optimizing inference in resource-limited settings.
\del{Canel et al. address resource-limited edge devices by reducing computational load and managing bandwidth. Yuan et al. introduce a learnable input filtering framework for mobile-centric inference, minimizing unnecessary data processing. Yuan et al. focus on accelerated decoding for real-time video analytics by selectively filtering frames. Mi et al. extend this by incorporating video quality adaptation for dynamic frame and quality adjustment. Han et al. propose deep compression techniques to reduce storage and computation in deep neural networks. Yuan et al. propose multi-stream packet gating to optimize large-scale concurrent video inference.}
\rev{They scale analytics on constrained nodes~\cite{canel2019scaling}, filter inputs~\cite{yuan2022infi}, accelerate decoding and adapt quality~\cite{yuan2023accdecoder, mi2024accelerated}, compress models~\cite{han2015deep}, or gate packets across streams~\cite{yuan2023packetgame}.}
However, when confronted with severe ABW drops at the last-mile bottleneck router, these server-side methods must endure prolonged frame delays before executing further optimizations for video analytics on the edge server.

\MyPara{Proactive packet discarding at routers.}
Router-based proactive packet discarding \cite{patil2019towards, sharma2014controlling, liu2018adaptive, carlucci2018controlling, hamadneh2012dynamic}  aims to reduce latency by preemptively dropping packets to prevent queue overflow and mitigate congestion.
\del{Patil et al. use fluid modeling to improve the understanding and analysis of the controlled delay algorithm in reducing network latency. Sharma et al. explore how the controlled delay algorithm helps mitigate the bufferbloat problem by controlling queue delay in internet traffic. Liu et al. present an adaptive AQM algorithm that employs a novel information compression model to optimize network congestion control. Carlucci et al. investigate the interaction between end-to-end and AQM algorithms for controlling queuing delays in real-time communications. Hamadneh et al. introduce a dynamic weight parameter for improving the performance of the random early detection algorithm in TCP networks.}
\rev{They model and tune CoDel~\cite{patil2019towards, sharma2014controlling}, compress queue-state information for AQM~\cite{liu2018adaptive}, study how AQM interacts with end-to-end RTC control~\cite{carlucci2018controlling}, or adapt RED's weights~\cite{hamadneh2012dynamic}.}
When applied to edge-based real-time video analytics, they lead to retransmission delays and frame-agnostic packet loss, and lack dynamic application-specific optimization.

%% file: 7-conclusion.tex
% \vspace{-.5em}
\mysection{Conclusion
\label{sec:conclusion}}
This paper presents \tool, a packet discarding technique to mitigate the impact of drastic ABW degradation at last-mile edge routers on real-time video analytics.
\SYS incorporates block information into RTP headers, evaluates the priority of each packet at the server side, and discards packets with low importance on edge routers upon sharp ABW drops. 
An extensive evaluation with real-world network traces and video datasets demonstrates that \SYS maintains low latency and high accuracy in edge-based real-time video analytics.

\section*{Acknowledgment}
We thank the anonymous reviewers for their comments.
This work was supported by an AWS AI PhD Fellowship, the National Science Foundation under grants CNS-2301343, CNS-2330831, CNS-2403254, IIS-2546642, and IIS-2551122, as well as an Alibaba Research Internship. 

%\vspace{3pt}
%\noindent {\em This work does not raise any ethical issues.}

%% file: Appendix_1.tex
%\clearpage

\section{Estimated number of RTP packets in a video frame}
\label{appendix:estimated_number_of_RTP_packets}
We show the estimated number of RTP packets per video frame in Figure \ref{fig:Design-3}.
\begin{figure}[ht]
    \centering
    \includegraphics[width=1\linewidth]{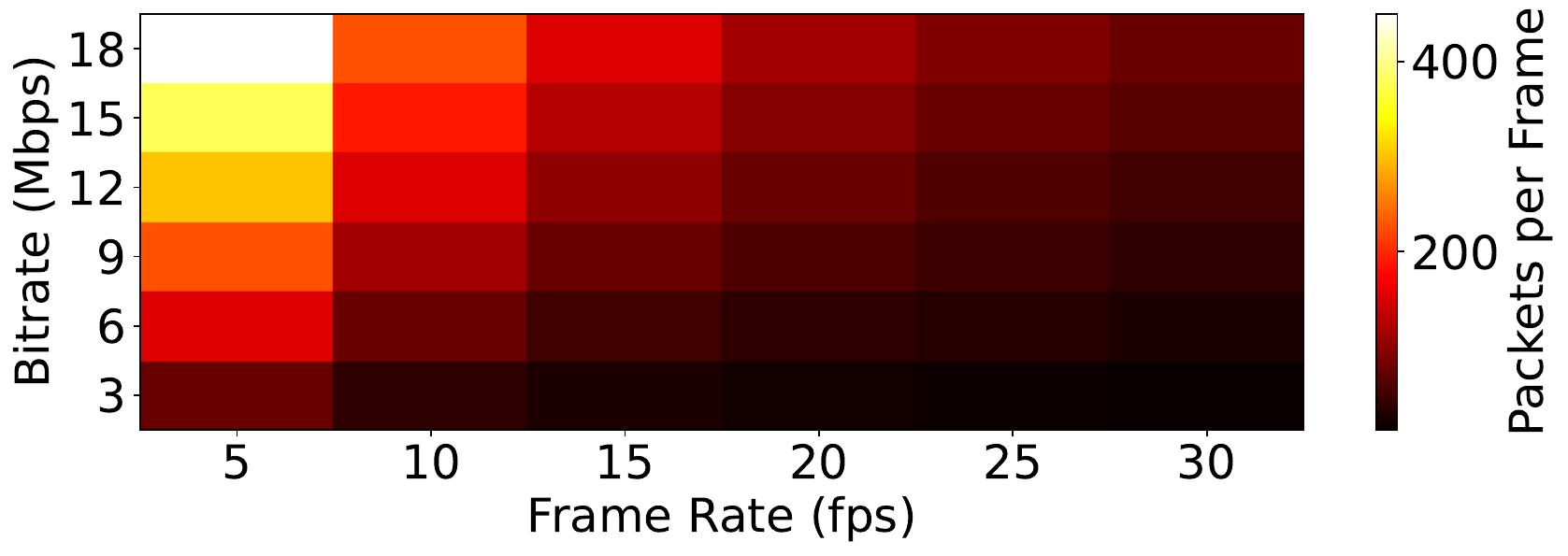}
    \vspace{0mm}
    \caption{Estimated number of RTP packets per video frame for common video bitrates and frame rates, assuming each packet is approximately 1000 bytes.}
    \vspace{0mm}
    \label{fig:Design-3}
\end{figure}

\section{Summary for block-aware RTP header extension}
\label{appendix:summary_for_rtp_extension}
We summarize the embedder and parser for each field with bit size in block-aware RTP header extension and limitations addressed in Table \ref{tab:Design-1}. 
Specifically, the limitations 1, 2 and 3 are abbreviated as $L_1$, $L_2$ and $L_3$, respectively.
\begin{table}[ht]
\caption{Embedder and parser for each field with bit size in block-aware RTP header extension and limitations addressed.}
    \footnotesize
    \centering
    \begin{tabular}{|c|cc|c|ccc|}
        \hline
        \textbf{\begin{tabular}[c]{@{}c@{}}Field\\Category\end{tabular}} & \multicolumn{2}{c|}{\begin{tabular}[c]{@{}c@{}}\textbf{Block-to-frame}\\\textbf{Metadata}\end{tabular}} & \begin{tabular}[c]{@{}c@{}}\textbf{Block}\\\textbf{Importance}\end{tabular} & \multicolumn{3}{c|}{\begin{tabular}[c]{@{}c@{}}\textbf{Retransmission}\\\textbf{Control}\end{tabular}} \\ 
        \hline
        
        \textbf{Field Name} & \multicolumn{1}{c|}{\textbf{S, E}} & \textbf{X, Y} & \textbf{CSR} & \multicolumn{1}{c|}{\textbf{C}} & \multicolumn{1}{c|}{\textbf{B}} & \textbf{Rem} \\ 
        \hline
        
        \textbf{Bit Size} & \multicolumn{1}{c|}{1, 1} & 1,1 & 7 & \multicolumn{1}{c|}{1} & \multicolumn{1}{c|}{1} & 9 \\ 
        \hline
        
        \textbf{Embedder} & \multicolumn{2}{c|}{Client} & Client & \multicolumn{2}{c|}{Router} & Client \\ 
        \hline
        
        \textbf{Parser} & \multicolumn{2}{c|}{Server} & Router & \multicolumn{3}{c|}{Server} \\ 
        \hline

        \textbf{\begin{tabular}[c]{@{}c@{}}Limitation\\Addressed\end{tabular}} & \multicolumn{2}{c|}{{\textbf{L$_1$}}} & {\textbf{L$_2$}} & \multicolumn{3}{c|}{{\textbf{L$_3$}}} \\ 
        \hline
    \end{tabular}
% \vspace{2mm}
\label{tab:Design-1}
% \vspace{4mm}
\end{table}

\section{Proposition \ref{thrm:DetectAbwDrop} for ABW drop detection}
\label{appendix:proof_of_proposition}
We present Proposition \ref{thrm:DetectAbwDrop} for ABW drop detection as follows.
\begin{proposition}
    \label{thrm:DetectAbwDrop}
    The failure to satisfy both of the two conditions in Equation (\ref{Equ:DetectAbwDrop_1}) could lead to a false detection of an ABW drop.
\end{proposition}
\begin{proof}
    Only satisfying the second condition in Equation (\ref{Equ:DetectAbwDrop_1}) indicates a decrease in the dequeue rate, which is a necessary but not sufficient condition for detecting an ABW drop.
    A counterexample, as shown in Figure \ref{fig:Design-6}, is the multiplicative decrease phrase, causing the enqueue rate of new coming video data to be lower than the dequeue rate, thereby failing to meet the first condition in Equation (\ref{Equ:DetectAbwDrop_1}). 
    Similarly, only satisfying the first condition in Equation (\ref{Equ:DetectAbwDrop_1}) is also insufficient, as it may merely reflect the additive increase phase.
    In summary, detecting a sudden and sharp drop in ABW requires both Equations in (\ref{Equ:DetectAbwDrop_1}) to be simultaneously satisfied.
\end{proof}

\begin{figure}[ht]
    \centering
    \includegraphics[width=1\linewidth]{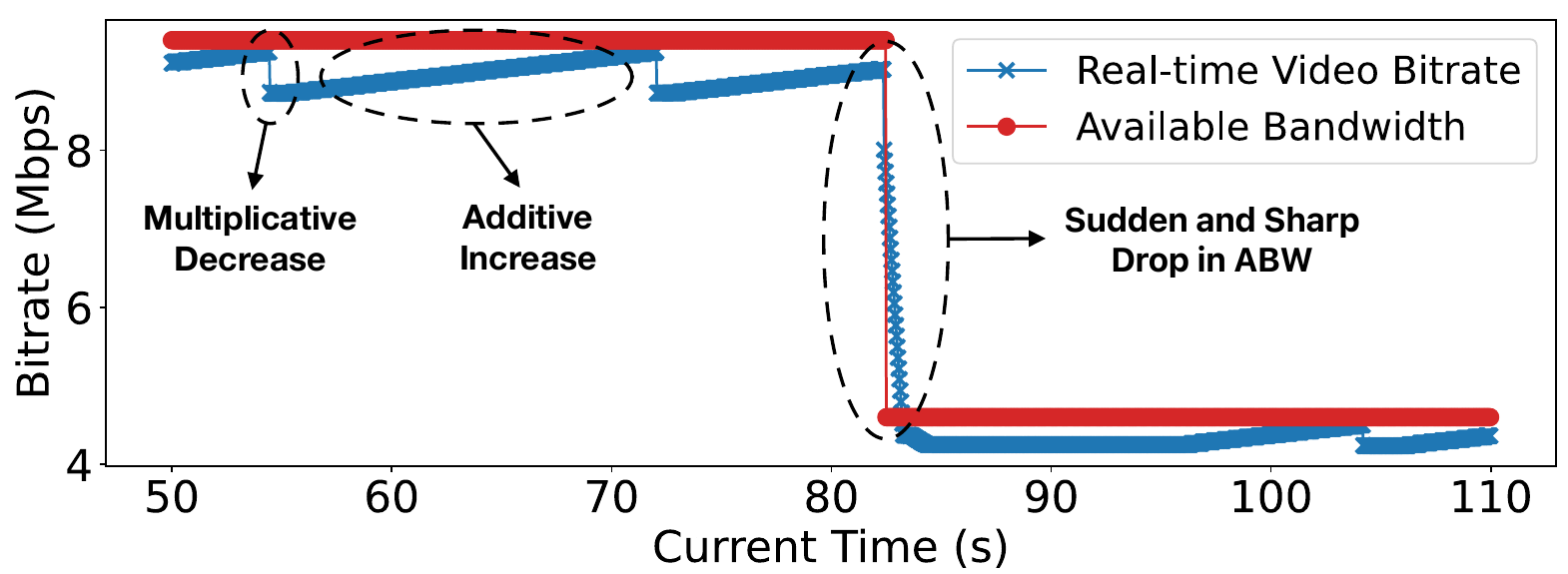}
    \vspace{0mm}
    \caption{Changes in real-time video bitrate before and after a sudden and sharp drop in available bandwidth.}
    % \vspace{2mm}
    \label{fig:Design-6}
\end{figure}

\section{Analysis for choice of $T_{win}$}
\label{appendix:analysis_for_t_win}
The choice of $T_{win}$ is guided by two primary considerations:
(a) \textit{Accuracy:} $T_{win}$ must not be too short to avoid the influence of random network fluctuations, which could result in erratic and inaccurate rate measurements;
(b) \textit{Responsiveness:} $T_{win}$ must not be excessively long, as this would delay the detection of a sudden and sharp decline in available bandwidth at the last-mile bottleneck router, thereby hindering timely packet discarding and potentially degrading real-time video analytics performance.
To address these constraints, $T_{win}$ is judiciously set to 100 ms, ensuring a balance between accuracy and responsiveness.

\section{Implementation of \SYS}
\label{appendix:detailed_implementation}

%Building on this, we show the workflow of \SYS in Figure \ref{fig:design-overview}.
Given that most last-mile bottleneck routers in edge networks are Linux-based \cite{Anurag, vieira2020fast} and support eBPF programs \cite{vieira2020fast}, we implemented phase-transition-based packet discarding with eBPF. Our program was applied to the egress path of network packets at the last-mile bottleneck router in our testbed through a new traffic control queue discipline.

%A more detailed implementation of \SYS can be found in Appendix \ref{appendix:detailed_implementation}.

We implemented the client based on the H.264 encoding \cite{H264} and RTP packetization in aiortc’s $\mathit{rtcrtpsender.\_run\_rtp()}$. 
In aiortc’s $\mathit{rtcrtpsender.\_next\_encoded\_frame()}$, each incoming frame is divided into $N\times$N blocks, where $N$ is a system parameter and each block individually encoded as an I/P-frame sequence with H.264 encoder.  
The block-level metadata for each frame is stored in $\mathit{framemarking\_list}$ using the custom $\mathit{FrameMarking}$ class, which is then passed back to $\mathit{rtcrtpsender.\_run\_rtp()}$, and embedded into the RTP header extension for transmission.

%Building on this, we show the workflow of \SYS in Figure \ref{fig:design-overview}.
%Each real-time video frame is first divided into multiple blocks of size $4\times4$, and each block is then encoded individually for subsequent RTP packetization.
%The block-aware RTP header extension augments the original RTP packets by adding block-to-frame metadata (including indicators for block start \textbf{S} and end \textbf{E}, and positional coordinates (\textbf{X}, \textbf{Y})), block priority \textbf{Pri}, and discarded block ranges (\textbf{Dis}, \textbf{Rem}).
%Details for RTP header extension are shown in Figure \ref{fig:Design-1}.

We implemented the $\mathit{rtp\_ext\_header}$ structure to support the parsing of block-level data from each RTP packet, and packets without the RTP header extension can pass through without any impact.
We used the $\mathit{bpf\_map\_def\,SEC("maps")}$ structure to record the current state, packet discarding priority, enqueue/dequeue rate, and other metrics. 
They are stored in the kernel's memory of the Linux-based router, allowing the eBPF program to access them with $\mathit{bpf\_map\_lookup\_elem()}$ when making packet discarding decisions for new packets.

%Subsequently, the RTP packets with extended headers are transmitted by the packet sender and then arrive at the bottleneck router.
%We modify the existing traffic control queue by adding a queue discipline that attaches an eBPF \cite{vieira2020fast} program to the egress path of network packets. 
%The state-transition-based packet discarding is implemented in the eBPF program to manage outgoing traffic. 
%Specifically, the eBPF program measures the enqueue and dequeue rates of the real-time video stream and makes appropriate packet discarding decisions by considering both the frame structure and network conditions.

We implemented the video analytics server in \SYS based on the H.264 decoding and RTP depacketization in aiortc’s $\mathit{rtcrtpreceiver}.\_\mathit{handle}\_\mathit{rtp}\_\mathit{packet}()$.
The block-level metadata in the RTP header extension is extracted by the modified function $\mathit{Jitterbuffer.add()}$, reconstructing the new frame from multiple blocks.

Building on this, we show the workflow of \SYS in Figure \ref{fig:design-overview}.
Each real-time video frame is first divided into multiple blocks of size $N\times N$, and each block is then encoded individually for subsequent RTP packetization.
The block-aware RTP header extension augments the original RTP packets by adding block-to-frame metadata (including indicators for block start \textbf{S} and end \textbf{E}, and positional coordinates (\textbf{X}, \textbf{Y})), block priority \textbf{Pri}, and discarded block ranges (\textbf{Dis}, \textbf{Rem}).
Details for RTP header extension are shown in Figure \ref{fig:Design-1}.

Subsequently, RTP packets with extended headers are transmitted by the packet sender and then arrive at the last-mile router.
We modified the existing traffic control queue by adding a queue discipline that attaches an eBPF \cite{vieira2020fast} program to the egress path of network packets. 
The phase-transition-based packet discarding is implemented in the eBPF program to manage outgoing traffic. 
Specifically, the eBPF program measures the enqueue and dequeue rates of the real-time video stream and makes appropriate packet discarding decisions by considering both the frame structure and network conditions.

Finally, the non-discarded RTP packets are received by the server, where they are then RTP-unpacketized and decoded into video frames using the block-aware information stored in the extended RTP headers.
The decoded frames are then processed by the video analytics model, and the analysis results are sent back to the client.
Additionally, priority-based video analytics feedback calculates the RAI of each block within a frame based on the analysis results, determining the block priority. 
Priorities are then written into our custom application-specific RTCP structure and returned to the client.
Details for custom RTCP structure are shown in Figure \ref{fig:Design-8}.

\begin{figure}[!t]
    \vspace{0mm}
    \centering
    \subfigbottomskip=1mm %设置第二行子图与第一行子图的距离，即下面的头与上面的脚的距离
    \subfigcapskip=0mm %设置子图与子标题之间的距离
    % 第零行
    \subfigure[Enqueue and dequeue rates.]{
        \includegraphics[width=0.46\linewidth]{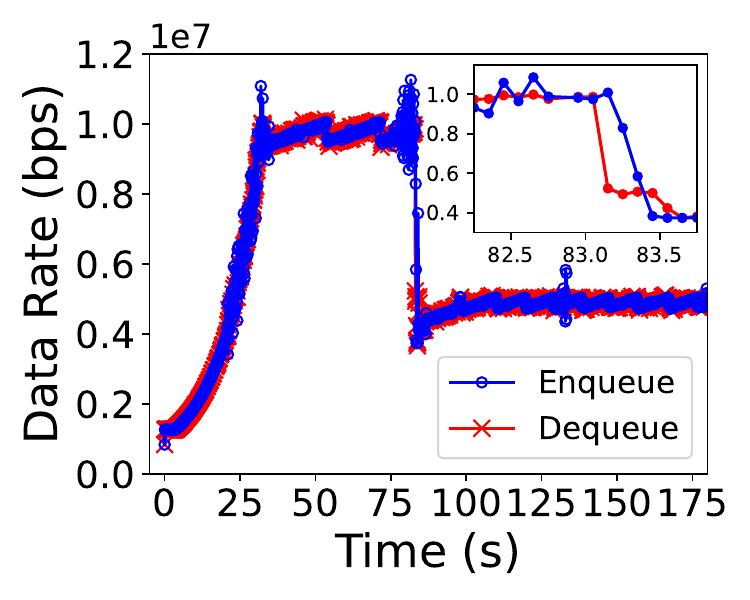}
        \label{fig:Eva_0_1}
    }\hfill
    \subfigure[Number of packets.]{
        \includegraphics[width=0.46\linewidth]{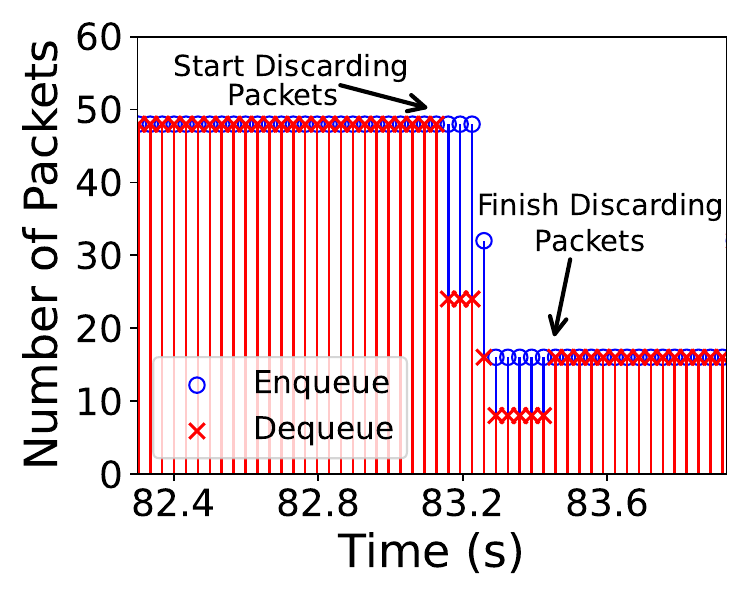}
        \label{fig:Eva_0_2}
    }\\
    % 第一行
    \subfigure[Frame delays with time.]{
        \includegraphics[width=0.46\linewidth]{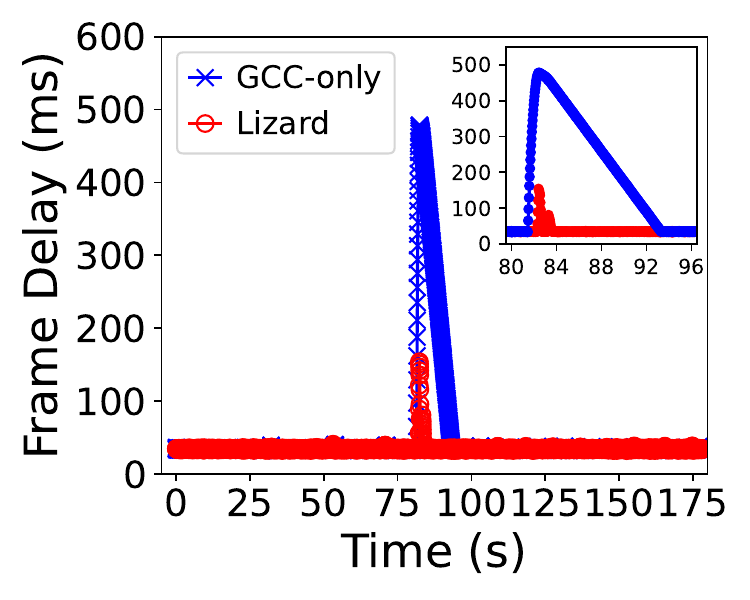}
        \label{fig:Eva_1}
    }\hfill
    \subfigure[Accuracy with time.]{
        \includegraphics[width=0.46\linewidth]{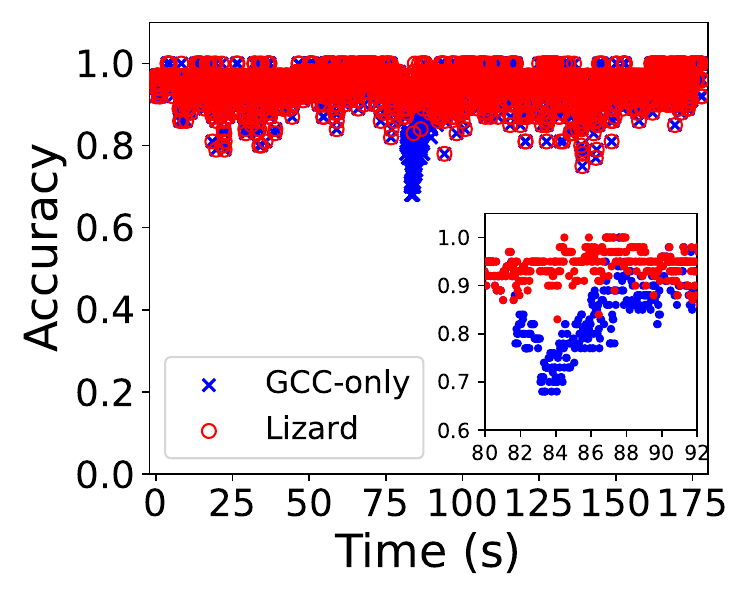}
        \label{fig:Eva_2}
    }\\
    % 第二行
    \subfigure[CDF for frame delays.]{
        \includegraphics[width=0.46\linewidth]{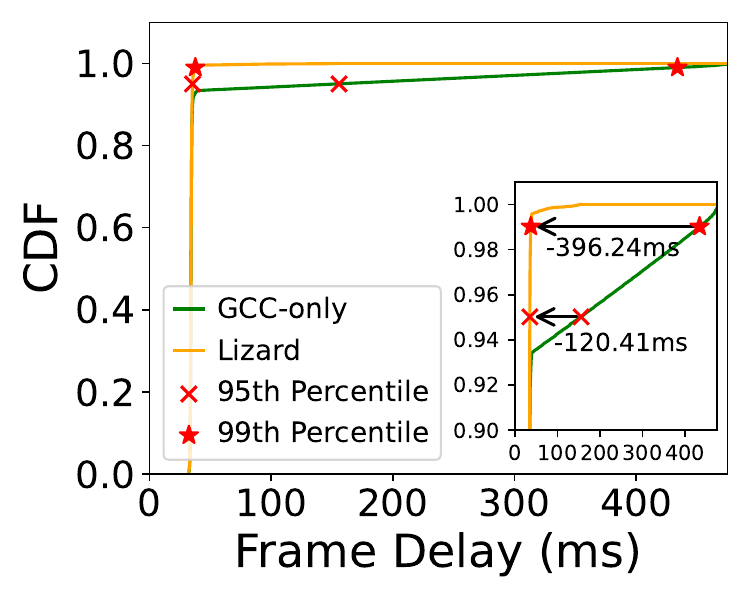}
        \label{fig:Eva_3}
    }\hfill
    \subfigure[CDF for accuracy.]{
        \includegraphics[width=0.46\linewidth]{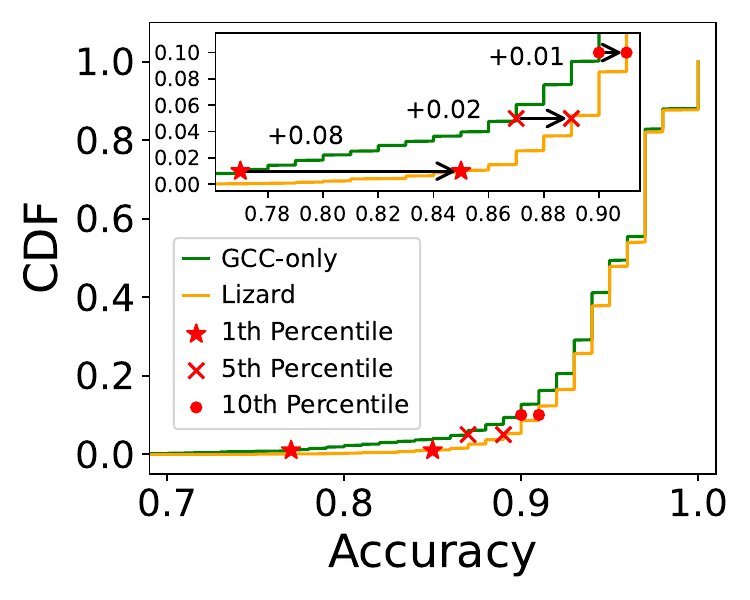}
        \label{fig:Eva_4}
    }
    \caption{A performance drill-down for \SYS.}
    \label{fig:Eva_0_1_2_3_4}
    \vspace{0mm}
\end{figure}

\section{Performance Drill Down}\label{sec:drilldown}
We performed a closer examination of \SYS's performance by demonstrating its effectiveness with a specific analytics program. 
In Figure \ref{fig:Eva_0_1_2_3_4}, we show how \SYS responds to a drastic ABW drop from 10 Mbps to 5 Mbps around 90 seconds. This includes the enqueue and dequeue rates, proactive packet discarding, frame delay, and accuracy variations.

As shown in Figure \ref{fig:Eva_0_1}, around 90 seconds, the dequeue rate of the considered video stream experiences a sharp drop from 9.9 Mbps to 5.2 Mbps due to the drastic ABW degradation.
However, the enqueue rate of the video stream remains stable at around 10 Mbps. 
Consequently, \SYS detects the drastic ABW degradation through $DetectAbwDrop$ in the \textbf{INIT} phase and transitions to the \textbf{PRE} phase.
Furthermore, once Lizard determines that the first packet of the next frame has been received, it enters the \textbf{PD} phase to initiate proactive packet discarding. 
As depicted in Figure \ref{fig:Eva_0_2}, approximately 50\% of the less-important packets are discarded in the \textbf{PD} phase.
Proactive packet discarding continues until \SYS detects ABW stabilization via $CheckAbwSteady$, i.e., when the dequeue and enqueue rates converge to an equilibrium.

Additionally, we present the temporal variations of frame delay and accuracy in Figures \ref{fig:Eva_1} and \ref{fig:Eva_2}. 
For a more intuitive understanding, we compared \SYS with the baseline GCC-only. 
Consistent with the analysis in \S\ref{sec:motivation}, GCC-only experiences significant lag in updating the bitrate and mitigating frame delay, resulting in an accuracy drop below 70\%. 
In contrast, \SYS can promptly mitigate queuing delays by proactively discarding packets upon detecting a drastic ABW drop at the last-mile router. 
Consequently, compared to GCC-only, \SYS reduces frame delay by 66.7\% and maintains accuracy at around 0.8.

To more clearly illustrate the extreme performance variations caused by the drastic ABW drop, we plot the CDFs of frame delays and accuracy for our 200-second running example. 
As shown in Figure \ref{fig:Eva_3}, we reduce the worst 0.1\% of frame delay from 480 milliseconds to approximately 120 milliseconds. 
In Figure \ref{fig:Eva_4}, we improve the worst 0.1\% of accuracy from 0.7 to around 0.8.
In summary, the above results validate the effectiveness of our proposed \SYS.